\def\REVISIONCLEAN{1}

\documentclass[aps,pre,superscriptaddress,nofootinbib,longbibliography, twocolumn]{revtex4-2}

\usepackage{amsmath,amssymb,amsfonts,amsthm,mathtools}
\usepackage{graphicx}
\usepackage{booktabs}
\usepackage{bm}
\usepackage{xcolor}
\usepackage{hyperref}
\graphicspath{{./}{figures/}}
\hypersetup{colorlinks=true,linkcolor=blue,citecolor=blue,urlcolor=blue}

\newcommand{\RefATwoColor}{\color{red}}
\ifdefined\REVISIONCLEAN
  \renewcommand{\RefATwoColor}{\color{black}}
\fi
\hypersetup{allcolors=black}

\expandafter\DeclareRobustCommand\csname HLRefA2\endcsname[1]{{\RefATwoColor #1}}
\long\def\HL_REF_#1#{\csname HLRef#1\endcsname}
\newcommand{\M}{\mathcal{M}}
\newcommand{\R}{\mathbb{R}}
\newcommand{\eps}{\varepsilon}
\newcommand{\Vol}{\operatorname{Vol}}
\newcommand{\dist}{\operatorname{dist}}
\newcommand{\diam}{\operatorname{diam}}
\newcommand{\Length}{\mathcal{L}}
\newcommand{\Tube}{\mathcal{T}}
\newcommand{\dd}{\mathrm{d}}

\newtheorem{definition}{Definition}
\newtheorem{theorem}{Theorem}
\newtheorem{corollary}{Corollary}
\newtheorem{proposition}{Proposition}
\newtheorem{remark}{Remark}

\begin{document}

\title{Finite-resolution exhaustive traversal of thermodynamic state spaces has divergent thermodynamic length}

\author{Satori Tsuzuki}
\email{tsuzukisatori@g.ecc.u-tokyo.ac.jp}
\affiliation{Research Center for Advanced Science and Technology, The University of Tokyo, Komaba 4-6-1, Meguro-ku, Tokyo, Japan}

\date{November 1, 2026}

\begin{abstract}
\HL_REF_B{Resolving a prescribed multidimensional control window by a single continuous, serial scan is a task distinct from driving between fixed endpoints.}
\HL_REF_AB{We formulate this task as metric coverage at resolution $\eps$. Classical covering estimates, applied to a compact regular $d$-dimensional Riemannian window, give a required length of order at least $\eps^{1-d}$; local grid traversals attain the same exponent. When a physical friction tensor uniformly dominates the coverage metric, this geometric cost implies a necessary duration of order at least $\eps^{-2(d-1)}$ at fixed quadratic excess-work budget. Fixed-duration divergence refers to the slow-driving quadratic action and is not a claim about exact work outside that regime.}
\HL_REF_A{A detailed-balance three-state process and an overdamped harmonic trap provide microscopic friction examples.}
\HL_REF_B{For the trap, independently integrated Gaussian moment dynamics approach the quadratic prediction as the prescribed raster is slowed, while refinement at a controlled work budget increases scan duration.}
\HL_REF_AB{A quasi-uniform operational resolution floor cuts off the geometric divergence. Landau-response and van der Waals proxy metrics serve separately as prescribed diagnostic geometries: independent quadrature tests their coordinate- and protocol-dependent critical prefactors without asserting universal critical dissipation. Morton ordering supplies an ordering/connector benchmark with larger finite-resolution amplitudes. The resulting resource statements apply to exhaustive serial coverage, with metric, interpolation and driving assumptions specified explicitly.}
\end{abstract}

\maketitle

\section{Introduction}

Thermodynamic states are often treated as points of a manifold of extensive variables, intensive variables, or externally controlled parameters.  On regular regions this manifold carries natural quadratic forms: Weinhold and Ruppeiner metrics from thermodynamic Hessians \cite{Weinhold1975,Ruppeiner1995}, Fisher--Rao metrics for equilibrium ensembles \cite{Rao1945,Amari2016}, and friction metrics controlling finite-time excess work \cite{SalamonBerry1983,Crooks2007,SivakCrooks2012,BlaberSivak2023,LoutchkoSughiyamaKobayashi2022,ScandiPerarnauLlobet2019}.  These metrics make distance and length operational rather than purely coordinate based.

Space-filling curves provide a complementary fact: a one-dimensional continuum can be mapped continuously and surjectively onto a higher-dimensional compact set \cite{Peano1890,Hilbert1891,Sagan1994}.  Their finite approximants are also useful orderings of multidimensional data \cite{Bader2013}.  Figure~\ref{fig:sfccurves} shows the finite-resolution viewpoint used here: Hilbert- and Peano-type approximants refine locally, while Morton/Z-order visits the same cells through a discrete ordering with nonlocal inter-block jumps.

The question addressed in this paper is what remains of space filling when the target is a thermodynamic state space with a thermodynamic metric.  The key distinction is that topological coverage and thermodynamic traversability are inequivalent.  An exact Peano or Hilbert curve may be continuous and surjective, but it is not a finite-length quasistatic thermodynamic protocol.  Conversely, any laboratory or simulation protocol has finite resolution and, once its interpolation is fixed, is represented by a rectifiable curve.  The relevant objects are therefore families \(H_\eps\) whose images are \(\eps\)-dense in the chosen thermodynamic distance.

\HL_REF_B{We consider a specified task: continuous, serial, exhaustive coverage of a prescribed control-space window.  A concrete example is mapping the equilibrium mean position and distribution width of a harmonically confined particle as its trap center and stiffness are varied.  Coverage then controls a spatial interpolation error, while equilibration, measurement precision, and temporal sample correlations require separate accounting.  The analysis applies when the task imposes exhaustive coverage; it does not assume that general inference or optimization tasks require it.}

\HL_REF_AB{The geometric coverage problem has close antecedents.  Miranda Jr., Paolini, and Stepanov \cite{MirandaPaoliniStepanov2006} study connected sets in Euclidean space that minimize the maximum distance to a target under a length constraint, together with the corresponding minimum-length problem at prescribed coverage radius.  Alvarado \emph{et al.} \cite{AlvaradoEtAl2023} investigate the length asymptotics of maximum-distance minimizers as that radius decreases.  These works measure the one-dimensional Hausdorff measure of the covering set.  Here the cost is parametrized traversal length, which also counts repeated passages.  Since \(\mathcal H_g^1(\gamma([0,1]))\le \Length_g[\gamma]\), a lower bound for the image also bounds the executed traversal; this distinction does not create a new coverage exponent.}

The main geometric statement is simple.  Let \((\M,g)\) be a compact regular \(d\)-dimensional thermodynamic state-space window with positive Riemannian volume.  If a rectifiable curve \(H_\eps\) comes within thermodynamic distance \(\eps\) of every point of \(\M\), then
\begin{equation}
    \Length_g[H_\eps] \gtrsim \Vol_g(\M)\,\eps^{1-d},
    \label{eq:intro_scaling}
\end{equation}
up to geometry-dependent constants and lower-order terms.  For \(d>1\), exhaustive coverage therefore requires divergent thermodynamic length as \(\eps\to0\).  Hilbert- and Peano-type grid traversals attain the same exponent on regular compact domains, so the power \(d-1\) is the codimension of a curve in a \(d\)-dimensional volume, not an artifact of a specific construction.

\HL_REF_AB{Our self-contained ball-covering proof places this classical mechanism in a compact regular Riemannian window with explicitly stated metric hypotheses.  It uses the same volume-growth principle as tube and Minkowski-content estimates \cite{Weyl1939,Gray2004,Mattila1995}.  We do not claim a new geometric exponent, an optimal covering constant, or a new solution of the Euclidean maximum-distance minimization problem.  The two- and three-dimensional traversal calculations are benchmarks of this known geometric scaling.}

\HL_REF_AB{The thermodynamic ingredients are also established: the friction-metric framework and its slow-driving inequality \(W_{\rm ex}^{(2)}\ge L_\zeta^2/\tau\) follow from linear response \cite{SivakCrooks2012}.  In particular, the two-control harmonic-trap friction tensor used below is already given in Eq.~(19) of Ref.~\cite{SivakCrooks2012}; its derivation here fixes conventions and supplies a solvable example.  The contribution of the present work is organized into three levels:}
\begin{enumerate}
    \item \HL_REF_AB{an explicit application of the geometric coverage estimate to compact regular thermodynamic windows, with coordinate resolution separated from intrinsic metric resolution;}
    \item \HL_REF_AB{the coverage-specific resource consequence obtained by combining that estimate with the established friction action: refinement constrains the duration of a serial scan at a prescribed quadratic excess-work budget;}
    \item \HL_REF_AB{operational realizations and checks in noncritical trap and Markov models, including a comparison of independently computed driven-trap work with the quadratic prediction, followed by exploratory tests of singular diagnostic geometries.}
\end{enumerate}

\HL_REF_AB{The resource statement is conditional.  If a physical friction tensor \(\zeta\) coincides with, or uniformly dominates, the coverage metric, then fixed-duration refinement gives \(\Omega(\eps^{2(1-d)}/\tau)\) within the quadratic action.  It does not establish the same divergence for exact physical work outside slow driving.  At fixed quadratic budget \(W_0\), the bound \(\tau\ge L_\zeta^2/W_0\) is a necessary duration condition; local relaxation times and the timing of connectors and direction changes must also permit slow driving.  The driven-trap calculations below test this regime rather than infer physical work from squaring the length alone.}

\HL_REF_AB{A finite operational resolution gives a cutoff law under the stated shape-regular, quasi-uniform partition assumptions.  For critical windows, a different issue is the cutoff dependence of the amplitude \(C_\Delta\) for a specified metric and traversal family after taking the fixed-window resolution limit.  The Landau-response Ising tensors are prescribed diagnostic geometries, and the van der Waals family is a controlled response proxy.  Their directional-integrability tests, supplemented by independent continuous quadrature, are not universal predictions of dissipation in physical critical systems, nor estimates of an intrinsic minimum covering prefactor.}

\HL_REF_AB{Section~\ref{sec:state_space} fixes the metric notation.  Section~\ref{sec:methods} states the geometric estimate and its conditional friction interpretation, gives the microscopic noncritical examples and operational assumptions, and then defines the diagnostic model geometries.  Section~\ref{sec:analysis} presents the resource-allocation and driven-trap tests, the geometric benchmarks, and the critical-prefactor and ordering/connector comparisons.  Sections~\ref{sec:discussion} and~\ref{sec:conclusion} discuss the physical scope.  Detailed grid settings, regression diagnostics, and pseudocode are provided in Appendices~\ref{app:numerical_details}--\ref{app:geometric_proof}.}

\begin{figure*}[t]
    \centering
    \includegraphics[width=0.9\textwidth]{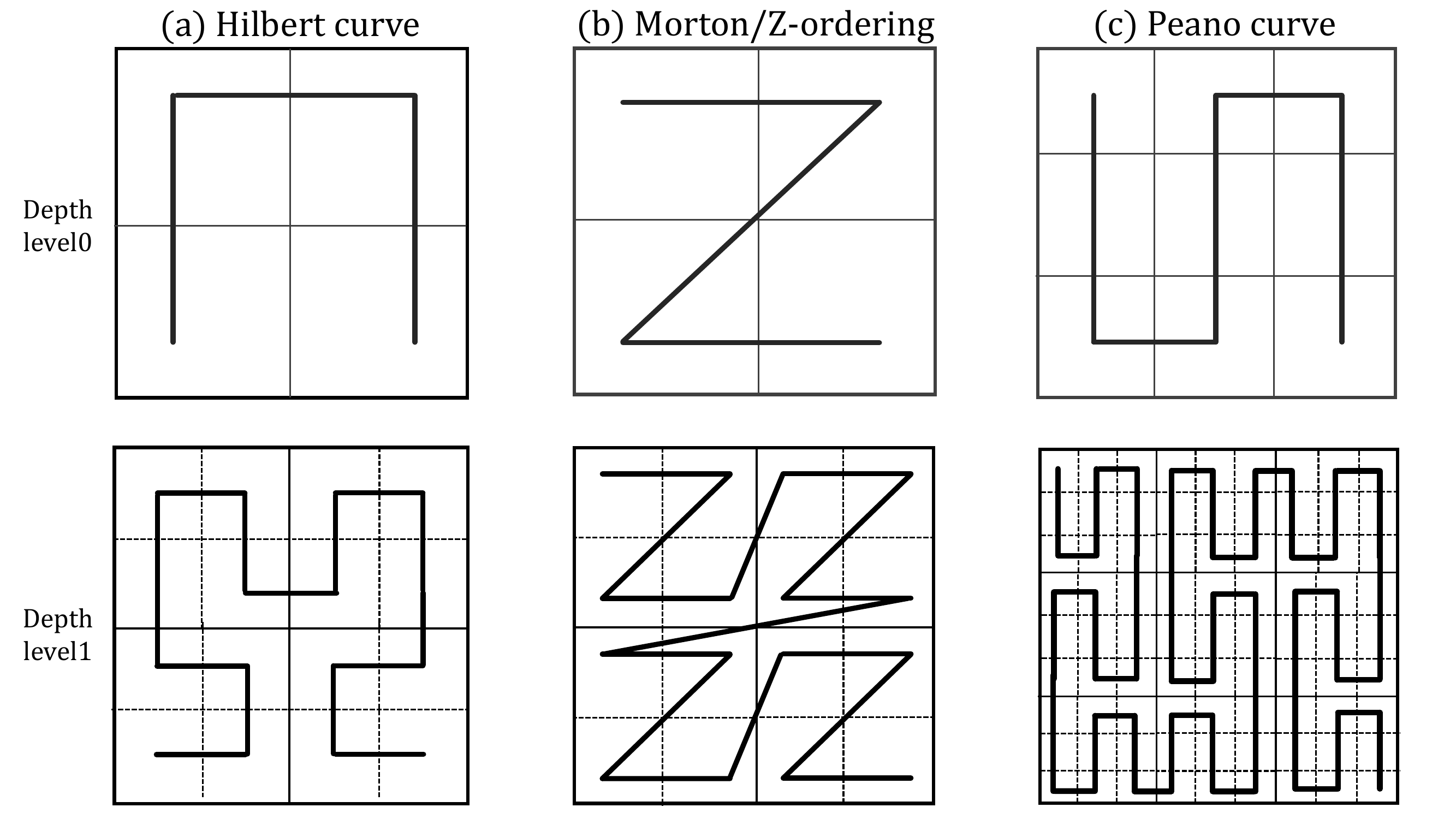}
    \caption{Schematic of finite-resolution space-filling traversals. Hilbert and Peano approximants refine the traversal locally, whereas \HL_REF_A{Morton/Z-order is a discrete cell ordering whose chosen connectors can introduce nonlocal inter-block jumps.}}
    \label{fig:sfccurves}
\end{figure*}

\section{Thermodynamic state spaces as Riemannian manifolds}
\label{sec:state_space}

Let \(\M\) denote the regular part of an equilibrium or control state space, coordinatized by extensive variables, intensive variables, or controls \(\lambda=(\lambda^1,\ldots,\lambda^d)\).  We assume throughout that \(\M\) is a connected compact \(d\)-dimensional smooth manifold, possibly with piecewise smooth boundary, with \(d>1\).  Compactness represents a finite experimental or numerical window after singular boundaries such as zero temperature, zero volume, spinodals, or critical points have been excluded; cutoff-dependent removals of these exclusions are treated below.

A thermodynamic metric is a smooth positive-definite tensor
\begin{equation}
    g = g_{ij}(x)\,\dd x^i\otimes \dd x^j
\end{equation}
with line element
\begin{equation}
    \dd \ell_g^2 = g_{ij}(x)\,\dd x^i\dd x^j.
    \label{eq:line_element}
\end{equation}
Examples include the Ruppeiner metric
\begin{equation}
    g^{\rm R}_{ij}(X)=-\frac{1}{k_B}\frac{\partial^2 S}{\partial X^i\partial X^j},
    \label{eq:ruppeiner_metric}
\end{equation}
on thermodynamically stable entropy-representation regions, the canonical Fisher metric
\begin{equation}
    g^{\rm F}_{ij}(\theta)=\partial_i\partial_j \psi(\theta),
    \qquad
    \psi(\theta)=\ln Z(\theta),
    \label{eq:fisher_metric}
\end{equation}
up to conventional system-size normalization, and the finite-time friction tensor \(\zeta_{ij}\), for which the excess power in linear response is
\begin{equation}
    P_{\rm ex}(t) = \dot\lambda^i(t)\,\zeta_{ij}(\lambda(t))\,\dot\lambda^j(t).
    \label{eq:excess_power}
\end{equation}
Under standard assumptions, \(\zeta\) is obtained from equilibrium time-correlation functions of generalized forces \cite{SivakCrooks2012}.

\HL_REF_B{Positive definiteness is an assumption to check for the specific controls and model, not an automatic consequence of the name ``Fisher'' or ``friction''. A covariance tensor can be singular, and a thermodynamic limit can remove an independent fluctuation direction. The geometric theorem below applies only after a positive-definite metric on the stated window has been specified. The Landau-response tensors in Sec.~\ref{subsubsec:ising} are prescribed diagnostic geometries rather than the thermodynamic-limit Curie--Weiss Fisher tensor.}

For an absolutely continuous curve \(\gamma:[0,1]\to\M\), the length induced by \(g\) is
\begin{equation}
    \Length_g[\gamma]
    =\int_0^1
    \left[g_{ij}(\gamma(t))\dot\gamma^i(t)\dot\gamma^j(t)\right]^{1/2}\dd t.
    \label{eq:length_def}
\end{equation}
The geodesic distance is \(\dist_g\), and the Riemannian volume is
\begin{equation}
    \Vol_g(A)=\int_A \dd\mu_g,
    \qquad
    \dd\mu_g=\sqrt{\det g(x)}\,\dd x^1\cdots \dd x^d .
    \label{eq:volume_def}
\end{equation}
We write
\begin{equation}
    \Tube_\eps(\Gamma)=\{x\in\M:\dist_g(x,\Gamma)\le \eps\}
    \label{eq:tube_def}
\end{equation}
for the closed \(\eps\)-neighborhood of \(\Gamma\subset\M\).

Metric comparison is used explicitly because the metric used to define coverage need not be the physical friction tensor.  On a compact regular window, smooth positive-definite metrics \(g\) and \(h\) are uniformly equivalent: there exist \(0<m_{gh}\le M_{gh}<\infty\) such that
\begin{equation}
    m_{gh}\,g(v,v)\le h(v,v)\le M_{gh}\,g(v,v)
    \label{eq:metric_equivalence}
\end{equation}
for all \(v\in T\M\).  Therefore
\begin{equation}
    \sqrt{m_{gh}}\,\Length_g[\gamma]
    \le \Length_h[\gamma]
    \le \sqrt{M_{gh}}\,\Length_g[\gamma].
    \label{eq:length_comparison}
\end{equation}
A coverage lower bound in \(g\) transfers to a friction length when \(\zeta(v,v)\ge m\,g(v,v)\) on the window.

Finally, an intrinsic resolution \(\eps_g\) and a coordinate mesh scale \(\delta\) are equivalent only up to metric-dependent constants.  On a compact coordinate window,
\begin{equation}
    a_g |x-y| \le \dist_g(x,y) \le b_g |x-y|
    \label{eq:coord_intrinsic_distance}
\end{equation}
for sufficiently small coordinate separations and some \(0<a_g\le b_g<\infty\).  Replacing \(\eps_g\) by \(\delta\) changes prefactors but not the resolution exponent.  On cutoff-dependent windows, the constants may depend on the cutoff and are included in the amplitude \(C_\Delta\).

\section{Methods}
\label{sec:methods}

\subsection{Finite-resolution space-filling parametrizations}
\label{subsec:finite_resolution}

An exact space-filling curve is a continuous surjection from an interval to a higher-dimensional compact set.  Such a map is useful as a limiting topological object, but it is not the appropriate representation of a finite thermodynamic protocol.  We therefore formulate space filling at finite resolution.

\begin{definition}[Finite-resolution thermodynamic space filling]
Let \((\M,g)\) be a compact thermodynamic state space.  A rectifiable curve \(H_\eps:[0,1]\to\M\) is called \(\eps\)-space-filling, or \(\eps\)-dense, if
\begin{equation}
    \M\subseteq \Tube_\eps\bigl(H_\eps([0,1])\bigr),
    \label{eq:eps_dense_def}
\end{equation}
that is, for every \(x\in\M\) there exists \(t\in[0,1]\) such that \(\dist_g(x,H_\eps(t))\le \eps\).
\end{definition}

The finite-resolution definition separates two notions that are often conflated.  The parameter interval provides an ordering of the states encountered by the protocol, whereas \(\eps\)-density specifies the spatial resolution with which the thermodynamic state space is explored.  The curve may self-intersect and need not be injective.  Self-intersections are physically natural because a protocol can return to states already visited; mathematically, allowing self-intersections only strengthens the lower bound below, since the proof does not rely on embeddedness.

Hilbert- and Peano-type approximants provide a canonical construction.  To make the scaling transparent, consider first a coordinate cube \(Q=[0,1]^d\) with a smooth metric \(g\) uniformly equivalent to the Euclidean metric.  Partition \(Q\) into \(m^d\) subcubes of coordinate side \(\delta=m^{-1}\).  A Hilbert traversal orders these subcubes so that consecutive subcubes share a face, except possibly at a finite number of chart-boundary connections in the manifold case.  Connecting the centers of consecutive subcubes by straight segments produces a piecewise linear curve \(H_m\).  Its coordinate covering radius is \(O(\delta)\), and Eq.~\eqref{eq:coord_intrinsic_distance} gives an intrinsic covering radius \(\eps_g=O(\delta)\).  Its length scales as
\begin{equation}
    \begin{aligned}
    \Length_g[H_m]&=O(m^d m^{-1})=O(m^{d-1})\\
                 &=O(\delta^{1-d})=O(\eps_g^{1-d}).
    \end{aligned}
    \label{eq:hilbert_upper_scaling}
\end{equation}
On a compact manifold covered by finitely many bi-Lipschitz coordinate charts, the same construction is obtained by performing such traversals in each chart and joining the chartwise paths by finitely many geodesic connectors.  The connector contribution is \(O(1)\) and does not affect the leading \(\eps_g^{1-d}\) scaling.

This construction shows that the scaling derived below is sharp at the level of exponents.  The lower bound says that no rectifiable curve can do parametrically better than \(\eps^{1-d}\), while Hilbert/Peano grid traversals show that this rate can be achieved up to metric-dependent constants on regular compact domains.

As a \HL_REF_A{finite-resolution ordering/connector benchmark}, we also consider Morton or Z-order at finite resolution \cite{Morton1966,Moon2001}.  On an order-\(n\) square grid, let each cell be indexed by integers \((i,j)\) with binary expansions \(i=\sum_{\ell=0}^{n-1}i_\ell 2^\ell\) and \(j=\sum_{\ell=0}^{n-1}j_\ell 2^\ell\).  A Morton key is obtained by bit interleaving,
\begin{equation}
    Z(i,j)=\sum_{\ell=0}^{n-1}
    \left(i_\ell 2^{2\ell}+j_\ell 2^{2\ell+1}\right),
    \label{eq:morton_key}
\end{equation}
and the cells are traversed in increasing \(Z\).  Unlike the Hilbert approximants used for the upper-bound construction, Morton order is not treated here as a continuous space-filling curve.  It is a discrete finite-resolution ordering that visits all cells but introduces nonlocal inter-block jumps.  It is therefore \HL_REF_A{useful as a specified ordering/connector benchmark for comparing resolution exponents and amplitudes.}

A physical protocol associated with such a discrete ordering is not specified by the ordering alone.  One must also choose connector paths \(\eta_n:[0,1]\to\M\) from \(x_n\) to \(x_{n+1}\).  The corresponding length would be
\begin{equation}
    \Length_g^{\rm conn}
    =\sum_{n=0}^{N-1}\int_0^1
    \left[g_{ij}(\eta_n(s))\dot\eta_n^i(s)\dot\eta_n^j(s)\right]^{1/2}\dd s .
    \label{eq:connector_length}
\end{equation}
For Hilbert-type local traversals, natural straight, geodesic, or midpoint-quadrature connectors differ only by lower-order or constant-factor effects on regular windows.  For Morton/Z-order, however, the nonlocal jumps make the connector choice part of the finite-resolution transition-cost model.  The Morton results below should therefore be read as a locality-control benchmark for a specified connector estimator, not as a connector-independent physical protocol.

For numerical implementation, one need not construct the singular limiting curve.  Let \(x_0,x_1,\ldots,x_N\) denote the ordered grid points or cell centers generated by a finite traversal, set \(\Delta x_n=x_{n+1}-x_n\), and let \(m_n=(x_n+x_{n+1})/2\).  For a coordinate-segment connector with local steps, the midpoint discretization is
\begin{equation}
    \Length_g^{\rm disc}
    =\sum_{n=0}^{N-1}
    \left[g_{ij}(m_n)\Delta x_n^i\Delta x_n^j\right]^{1/2},
    \label{eq:discrete_length}
\end{equation}
with quadrature error
\begin{equation}
    \Length_g^{\rm conn}-\Length_g^{\rm disc}
    =O\!\left(\sum_{n=0}^{N-1}\|\Delta x_n\|^3\right)
    \label{eq:discrete_length_error}
\end{equation}
when \(g\) is \(C^2\) on the window and \(\max_n\|\Delta x_n\|\to0\).  Thus a local Hilbert/Peano traversal on a \(d\)-dimensional grid with mesh \(\delta\) has absolute quadrature error \(O(\delta^{3-d})\), while its leading length is \(O(\delta^{1-d})\); the relative error is \(O(\delta^2)\) and does not change the finite-resolution exponent.  If an endpoint rather than a midpoint estimator is used, the accumulated local-discretization error is generally \(O(\sum_n\|\Delta x_n\|^2)\); this also leaves the leading exponent unchanged but gives slower convergence.

The theoretical analysis concerns the asymptotic behavior of finite-resolution curves as \(\eps\to0\), not the singular limiting curve itself.  For Morton/Z-order, which contains nonlocal inter-block jumps, Eq.~\eqref{eq:discrete_length} is used as a discrete transition-cost estimator corresponding to the chosen coordinate-segment connector; the small-step quadrature estimate \eqref{eq:discrete_length_error} is not invoked for those nonlocal jumps.  If a different physical interpolation is intended, Eq.~\eqref{eq:connector_length}, or the sum of geodesic distances \(\sum_n\dist_g(x_n,x_{n+1})\), should be used instead.

\subsection{Lower bound on thermodynamic length}
\label{subsec:length_lower_bound}

The lower bound is a metric covering statement.  It is independent of the microscopic origin of the metric and depends only on the Riemannian volume growth of small balls.  Since \(\M\) is compact and \(g\) is smooth, there exist constants \(r_0>0\) and \(C_{\rm ball}<\infty\) such that
\begin{equation}
    \Vol_g\bigl(B_g(x,r)\bigr)\le C_{\rm ball} r^d
    \label{eq:ball_growth}
\end{equation}
for all \(x\in\M\) and all \(0<r\le r_0\).  Here \(B_g(x,r)\) is the geodesic ball of radius \(r\).  Equation~\eqref{eq:ball_growth} follows from compactness and smoothness; if desired, sharper constants can be obtained from standard tube-volume expansions \cite{Weyl1939, Gray2004, Mattila1995}.

\begin{theorem}[Length cost of finite-resolution thermodynamic space filling]
\label{thm:length_lower_bound}
Let \((\M,g)\) be a compact connected \(d\)-dimensional Riemannian thermodynamic state space with \(d>1\) and \(V_g=\Vol_g(\M)>0\).  Assume the ball-growth estimate \eqref{eq:ball_growth}.  If an absolutely continuous curve \(H_\eps:[0,1]\to\M\) is \(\eps\)-space-filling and \(2\eps\le r_0\), then
\begin{equation}
    \Length_g[H_\eps]
    \ge
    \frac{V_g}{2^d C_{\rm ball}}\,\eps^{1-d}-2\eps.
    \label{eq:length_lower_bound}
\end{equation}
Consequently,
\begin{equation}
    \liminf_{\eps\to0}\,\eps^{d-1}\Length_g[H_\eps]
    \ge \frac{V_g}{2^d C_{\rm ball}},
    \label{eq:liminf_bound}
\end{equation}
so \(\Length_g[H_\eps]\to\infty\) for \(d>1\).
\end{theorem}

\begin{proof}[\HL_REF_AB{Proof sketch}]
\HL_REF_AB{An arclength-$\eps$ net on the curve has at most $L_g/\eps+2$ points, and its radius-$2\eps$ balls cover the required window. Applying Eq.~\eqref{eq:ball_growth} gives the bound. The full classical covering argument and the relation to sharper tube constants are collected in Appendix~\ref{app:geometric_proof}.}
\end{proof}

\HL_REF_AB{The estimate is used at the level of its exponent; no optimal tube or covering constant is claimed.}

\begin{corollary}[No finite-length exact filling of regular state-space volume]
Let \(\{H_{\eps_n}\}\) be a sequence of rectifiable curves with \(\eps_n\to0\) that becomes dense in \(\M\).  If \(d>1\) and \(\Vol_g(\M)>0\), then \(\sup_n\Length_g[H_{\eps_n}]=\infty\).  Thus an exact continuous space-filling parametrization can only arise as a singular limit of curves whose thermodynamic length diverges.
\end{corollary}

\HL_REF_AB{The corollary separates a topological parametrization from finite-length traversal in the specified metric. A physical-work consequence additionally requires the friction hypotheses below. The geometric conclusion is stable under regular coordinate changes and uniformly equivalent metrics on a fixed compact window.}

\subsection{Dissipation interpretation through thermodynamic length}
\label{subsec:dissipation}

Let \(\lambda(s)\in\M\), \(s\in[0,\tau]\), be a protocol executed in physical time \(\tau\).  In the slow-driving linear-response regime, the leading excess work is the quadratic friction action
\begin{equation}
    W_{\rm ex}^{(2)}[\lambda]
    =\int_0^\tau
    \dot\lambda^i(s)\zeta_{ij}(\lambda(s))\dot\lambda^j(s)\,\dd s,
    \label{eq:wex_def}
\end{equation}
where convention-dependent factors such as \(\beta\) or \(1/2\) are absorbed into \(\zeta\).  The associated friction length is
\begin{equation}
    \Length_\zeta[\lambda]
    =\int_0^\tau
    \left[\dot\lambda^i(s)\zeta_{ij}(\lambda(s))\dot\lambda^j(s)\right]^{1/2}\dd s.
    \label{eq:friction_length}
\end{equation}
Cauchy--Schwarz gives
\begin{equation}
    W_{\rm ex}^{(2)}[\lambda]
    \ge \frac{\Length_\zeta[\lambda]^2}{\tau},
    \label{eq:cs_dissipation}
\end{equation}
with equality in the quadratic action for constant friction speed.

\HL_REF_AB{Equation~\eqref{eq:cs_dissipation} is exact for the quadratic action, not a bound on the exact work of an arbitrarily fast driven process. At fixed duration, refinement makes the typical friction speed grow without bound and can invalidate the slow-driving approximation. The fixed-duration divergence below is therefore an action-level result. The operational interpretation instead concerns a family of protocols whose duration increases while their local driving rates remain slow compared with relaxation.}

Assume that coverage is defined by \(g\) and dissipation by a physical friction tensor \(\zeta\) that uniformly dominates \(g\),
\begin{equation}
    \zeta(v,v)\ge m_{g\zeta}g(v,v),
    \qquad m_{g\zeta}>0 .
    \label{eq:zeta_g_domination}
\end{equation}
Then
\begin{equation}
    \Length_\zeta[H_\eps]\ge \sqrt{m_{g\zeta}}\,\Length_g[H_\eps].
    \label{eq:zeta_g_comparison}
\end{equation}
\HL_REF_B{To state the finite-resolution bound without squaring a possibly negative lower estimate, define
\begin{equation}
 B(\eps)=\frac{V_g}{2^d C_{\rm ball}}\eps^{1-d}-2\eps,
 \qquad [B]_+=\max\{0,B\}.
 \label{eq:positive_length_bound}
\end{equation}
Combining the metric comparison, Theorem~\ref{thm:length_lower_bound}, and Eq.~\eqref{eq:cs_dissipation} then gives}
\begin{equation}
    W_{\rm ex}^{(2)}[H_\eps]
    \ge
    \HL_REF_B{\frac{m_{g\zeta}}{\tau}\,[B(\eps)]_+^2}.
    \label{eq:dissipation_lower_bound}
\end{equation}
Thus, within the quadratic friction action,
\begin{equation}
    W_{\rm ex}^{(2)}[H_\eps]=\Omega\left(\frac{\eps^{2(1-d)}}{\tau}\right),
    \qquad \eps\to0.
    \label{eq:wex_scaling}
\end{equation}
\HL_REF_AB{Equivalently, imposing a quadratic excess-work budget \(W_{\rm ex}^{(2)}\le W_0\) requires}
\begin{equation}
    \tau
    \ge
    \HL_REF_B{\frac{m_{g\zeta}}{W_0}\,[B(\eps)]_+^2}
    =\Omega\left(\eps^{2(1-d)}\right).
    \label{eq:time_budget}
\end{equation}

\HL_REF_AB{The duration inequality is necessary within the quadratic description; it does not by itself guarantee slow driving. Its physical use requires specified controls and connectors, and driving rates and direction changes compatible with the local relaxation times. Near criticality these times may themselves grow with the cutoff. A prescribed positive-definite tensor without an independently justified dynamical origin supplies a coverage geometry, not a microscopic dissipation prediction.}

\subsection{Microscopically derived friction metric for an overdamped harmonic trap}
\label{subsec:harmonic_trap_friction}

\HL_REF_AB{For a specified stochastic dynamics, the coverage result can be combined with a physical friction tensor. We use the established overdamped harmonic-trap tensor of Sivak and Crooks~\cite{SivakCrooks2012}; the derivation below fixes our conventions, rather than supplying a new friction formula. The example is a Brownian particle in a one-dimensional harmonic trap with controllable center and stiffness,}
\begin{equation}
    U(x;a,k)=\frac{k}{2}(x-a)^2,
    \qquad \lambda=(a,k),\qquad k>0 .
    \label{eq:harmonic_potential}
\end{equation}
The overdamped Langevin dynamics at inverse temperature \(\beta\) and mobility \(\mu\) is
\begin{equation}
    \dot x_t=-\mu k(x_t-a)+\sqrt{2\mu/\beta}\,\eta_t,
    \label{eq:overdamped_langevin}
\end{equation}
with \(\eta_t\) normalized white noise.  In the slow-driving linear-response regime, the friction tensor is the equilibrium time integral of the generalized-force covariance \cite{SivakCrooks2012,Zulkowski2012},
\begin{equation}
    \zeta_{ij}(\lambda)
    =\beta\int_0^\infty
    \left\langle \delta X_i(t)\,\delta X_j(0)\right\rangle_\lambda\,\dd t,
    \qquad
    X_i=-\partial_{\lambda^i}U .
    \label{eq:force_correlation_friction}
\end{equation}
Here \(\delta X_i=X_i-\langle X_i\rangle_\lambda\), and the average is taken in the equilibrium state at fixed \(\lambda\).

Writing \(y=x-a\), the stationary process is Ornstein--Uhlenbeck with \(\langle y(t)y(0)\rangle=(\beta k)^{-1}e^{-\mu kt}\).  The generalized forces are
\begin{equation}
    X_a=k y,
    \qquad
    X_k=-\frac{1}{2}y^2 .
    \label{eq:harmonic_generalized_forces}
\end{equation}
Gaussian moment identities give
\begin{align}
    \left\langle \delta X_a(t)\delta X_a(0)\right\rangle
    &=\frac{k}{\beta}e^{-\mu kt},\\
    \left\langle \delta X_a(t)\delta X_k(0)\right\rangle
    &=0,\\
    \left\langle \delta X_k(t)\delta X_k(0)\right\rangle
    &=\frac{1}{2\beta^2k^2}e^{-2\mu kt} .
    \label{eq:harmonic_force_covariances}
\end{align}
Substitution into Eq.~\eqref{eq:force_correlation_friction} yields the microscopic friction tensor
\begin{equation}
    \zeta_{aa}=\frac{1}{\mu},
    \qquad
    \zeta_{ak}=0,
    \qquad
    \zeta_{kk}=\frac{1}{4\beta\mu k^3},
    \label{eq:harmonic_friction_components}
\end{equation}
or equivalently
\begin{equation}
    \dd \ell_\zeta^2
    =\frac{1}{\mu}\,\dd a^2
    +\frac{1}{4\beta\mu k^3}\,\dd k^2 .
    \label{eq:harmonic_friction_metric}
\end{equation}
On every compact control window
\begin{equation}
    \M_{\mathrm{trap}}=[a_-,a_+]\times[k_-,k_+],
    \qquad 0<k_-<k_+<\infty,
    \label{eq:harmonic_control_window}
\end{equation}
this tensor is a smooth positive-definite Riemannian metric.  Taking the coverage metric to be \(g=\zeta\), Theorem~\ref{thm:length_lower_bound} applies directly and no auxiliary comparison between a diagnostic geometry and a friction tensor is needed.  For any intrinsically \(\eps\)-dense traversal of the two-dimensional trap-control window,
\begin{equation}
    \Length_\zeta[H_\eps]\ge C_{\rm trap}\eps^{-1}-O(\eps),
    \qquad
    W_{\rm ex}^{(2)}[H_\eps]\ge \frac{\Length_\zeta[H_\eps]^2}{\tau},
    \label{eq:harmonic_space_filling_bound}
\end{equation}
with \(C_{\rm trap}>0\) depending on the chosen compact window. \HL_REF_AB{The dynamical origin of this metric permits the conditional slow-driving interpretation; it does not extend the quadratic bound to arbitrary finite-time work.}

\HL_REF_B{A concrete coverage task is to map the equilibrium mean and width of the trapped distribution over this control window. With \(\sigma=(\beta k)^{-1/2}\), Eq.~\eqref{eq:harmonic_friction_metric} becomes
\begin{equation}
 \dd\ell_\zeta^2=\frac{1}{\mu}\bigl(\dd a^2+\dd\sigma^2\bigr).
 \label{eq:trap_observable_accuracy}
\end{equation}
Consequently an intrinsically \(\eps\)-dense set of recorded equilibrium responses controls the nearest-response interpolation errors in both the mean \(a\) and width \(\sigma\) by \(\sqrt{\mu}\eps\). This is the geometric component of the accuracy of a mean--width map. Statistical measurement error and non-equilibrium bias must be controlled separately; continuous passage through a cell alone does not supply an accurate independent measurement.}

\subsection{Microscopic friction from a detailed-balance Markov jump process}
\label{subsec:markov_jump_friction}

The harmonic trap above is a continuous Langevin example.  A finite-state master equation gives an equally explicit microscopic realization and shows that the friction metric is not tied to diffusive dynamics.  Consider a continuous-time Markov chain on three states $n=0,1,2$ with controllable energies
\begin{equation}
    U_0=0,\qquad U_1=\lambda^1,\qquad U_2=\lambda^2,
    \label{eq:three_state_energies}
\end{equation}
and equilibrium probabilities
\begin{equation}
    \pi_n(\lambda)=\frac{e^{-\beta U_n(\lambda)}}{Z(\lambda)},
    \qquad
    Z=1+e^{-\beta\lambda^1}+e^{-\beta\lambda^2} .
    \label{eq:three_state_pi}
\end{equation}
At fixed $\lambda$, choose the heat-bath jump rates
\begin{equation}
    w_{n\to m}(\lambda)=\gamma\,\pi_m(\lambda),
    \qquad m\ne n,
    \label{eq:heat_bath_rates}
\end{equation}
with $\gamma>0$.  These rates satisfy detailed balance, $\pi_n w_{n\to m}=\pi_m w_{m\to n}$, and make $\pi$ the unique stationary distribution.  For any observable $f(n)$, the backward generator acts as
\begin{equation}
    (\mathcal L_\lambda f)(n)
    =\gamma\left(\langle f\rangle_\pi-f(n)\right).
    \label{eq:heat_bath_generator}
\end{equation}
Therefore every zero-mean observable relaxes with the single rate $\gamma$:
\begin{equation}
    \left\langle f(t)g(0)\right\rangle_\lambda
    =e^{-\gamma t}\left\langle f g\right\rangle_\lambda
    \quad
    \text{when }\langle f\rangle_\pi=0 .
    \label{eq:heat_bath_correlation_decay}
\end{equation}

The generalized forces conjugate to the energy controls are
\begin{equation}
    X_i=-\partial_{\lambda^i}U=-\mathbf{1}_{n=i},
    \qquad i=1,2 .
    \label{eq:three_state_generalized_forces}
\end{equation}
Using the same correlation formula as Eq.~\eqref{eq:force_correlation_friction}, Eq.~\eqref{eq:heat_bath_correlation_decay} gives
\begin{align}
    \zeta_{ij}(\lambda)
    &=\beta\int_0^\infty
      \left\langle \delta X_i(t)\delta X_j(0)\right\rangle_\lambda\,\dd t \\
    &=\frac{\beta}{\gamma}
      \left(\pi_i\delta_{ij}-\pi_i\pi_j\right),
    \qquad i,j=1,2 .
    \label{eq:three_state_friction_tensor}
\end{align}
Equivalently,
\begin{align}
    \dd\ell_\zeta^2
    =\frac{\beta}{\gamma}\bigl[&\pi_1(1-\pi_1)(\dd\lambda^1)^2
    -2\pi_1\pi_2\,\dd\lambda^1\dd\lambda^2 \\
    &+\pi_2(1-\pi_2)(\dd\lambda^2)^2\bigr].
    \label{eq:three_state_metric}
\end{align}
The determinant is
\begin{equation}
    \det \zeta
    =\left(\frac{\beta}{\gamma}\right)^2\pi_0\pi_1\pi_2,
    \label{eq:three_state_metric_determinant}
\end{equation}
so the tensor is positive definite whenever all three equilibrium probabilities are nonzero.  On any compact two-control window $\lambda^i\in[\lambda_-^i,\lambda_+^i]$, the probabilities are bounded away from zero and Eq.~\eqref{eq:three_state_friction_tensor} is a smooth Riemannian friction metric.  Taking $g=\zeta$, Theorem~\ref{thm:length_lower_bound} and Eq.~\eqref{eq:cs_dissipation} therefore imply
\begin{equation}
    \Length_\zeta[H_\eps]\ge C_{\rm jump}\eps^{-1}-O(\eps),
    \qquad
    W_{\rm ex}^{(2)}[H_\eps]\ge\frac{\Length_\zeta[H_\eps]^2}{\tau},
    \label{eq:three_state_space_filling_bound}
\end{equation}
for any intrinsically $\eps$-dense traversal of the two-dimensional energy-control window.  Since the metric components are uniformly bounded above and below on such a window, any local raster with mesh spacing $\delta$ has $\Length_\zeta=\Theta(\delta^{-1})$ and a quadratic-action fixed-duration lower bound $W_{\rm ex}^{(2)}=\Omega(\delta^{-2}/\tau)$.  This example supplies a fully microscopic master-equation counterpart to the Langevin trap: the friction tensor follows from the specified transition rates and not from an imposed diagnostic geometry.

\subsection{Quantitative operational sweep cost}
\label{subsec:operational_sweep_cost}

The harmonic trap provides, in addition, a concrete operational primitive for exhaustive control-space sweeps.  Consider a serpentine raster of the window \eqref{eq:harmonic_control_window}, with rows parallel to the trap-center coordinate \(a\) and coordinate row spacing \(\delta_k\) in stiffness.  Let
\begin{equation}
    A=a_+-a_-,
    \qquad
    K=k_+-k_- .
    \label{eq:trap_window_lengths}
\end{equation}
At each fixed \(k\), one full row has friction length \(\int_{a_-}^{a_+}\sqrt{\zeta_{aa}}\,\dd a=A/\sqrt{\mu}\).  The number of rows is \(K/\delta_k+O(1)\), so the horizontal contribution is
\begin{equation}
    \Length_{a,\mathrm{raster}}
    =\frac{A K}{\sqrt{\mu}\,\delta_k}+O(1).
    \label{eq:raster_horizontal_length}
\end{equation}
The vertical connectors between rows contribute only a bounded term,
\begin{align}
    \Length_{k,\mathrm{raster}}
    &=\int_{k_-}^{k_+}\sqrt{\zeta_{kk}}\,\dd k+O(\delta_k) \\
    &=\frac{1}{\sqrt{\beta\mu}}
      \left(k_-^{-1/2}-k_+^{-1/2}\right)+O(\delta_k).
    \label{eq:raster_vertical_length}
\end{align}
Therefore
\begin{equation}
    \Length_{\zeta,\mathrm{raster}}
    =\frac{A K}{\sqrt{\mu}\,\delta_k}+O(1),
    \label{eq:raster_total_length}
\end{equation}
and, within the same quadratic slow-driving description, any realization of this raster in duration \(\tau\) satisfies
\begin{equation}
    W_{\rm ex}^{(2)}
    \ge
    \frac{1}{\tau}
    \left(
    \frac{A K}{\sqrt{\mu}\,\delta_k}+O(1)
    \right)^2
    =\Omega\!\left(\frac{A^2K^2}{\mu\tau\,\delta_k^2}\right).
    \label{eq:raster_work_lower_bound}
\end{equation}
\HL_REF_AB{Halving the stiffness-row spacing therefore increases the leading quadratic action by a factor of four at fixed duration. At fixed quadratic budget it instead requires a duration increasing with the same power, together with a separate check that the timing remains slow. This noncritical example gives a resource consequence for a prescribed serial sweep of a whole control window. It does not assert the same requirement for an inference or optimization task that can avoid exhaustive coverage.}

\subsection{Operational resolution cutoffs and discrete-state traversal}
\label{subsec:operational_resolution_cutoff}

The divergence in Theorem~\ref{thm:length_lower_bound} is a continuum statement. It should not be read as an infinite cost at any fixed laboratory or numerical resolution. \HL_REF_B{For the cutoff law, assume a shape-regular quasi-uniform partition \(\mathcal P_\Delta\) of the prescribed regular window: cell diameters and inscribed-ball radii are comparable to one common scale, representatives are separated by comparable distances, and the local adjacency graph is connected with bounded degree. The comparison constants are independent of refinement on this fixed window. Maximum diameter alone would not exclude arbitrarily finer subregions and is insufficient for the cell-count and traversal estimates below.} Define the intrinsic mesh diameter by
\begin{equation}
    \Delta_g=\sup_{C\in\mathcal P_\Delta}\diam_g(C).
    \label{eq:intrinsic_resolution_floor}
\end{equation}
If two points lie in the same cell, a protocol that distinguishes them is below the operational resolution of the apparatus or simulation.  The meaningful covering radius is therefore not \(\eps\) alone but
\begin{equation}
    r_{\rm op}=\max\{\eps,\Delta_g\}.
    \label{eq:effective_operational_radius}
\end{equation}
On a regular compact \(d\)-dimensional window, the same lower-bound and Hilbert/Peano upper-bound arguments give the cutoff-controlled scaling
\begin{equation}
    \Length_{\rm op}^{*}(\eps;\Delta_g)
    =\Theta\!\left(r_{\rm op}^{1-d}\right)
    =\Theta\!\left(\max\{\eps,\Delta_g\}^{1-d}\right),
    \label{eq:operational_length_scaling}
\end{equation}
\HL_REF_B{Here \(\Length_{\rm op}^{*}\) denotes the minimum leading-order length for local exhaustive coverage at the effective scale on this regular partition family, up to window- and partition-dependent constants. The lower estimate follows from coverage; a local walk through the cells, allowing repeated edges, gives the upper scaling.} Thus finite resolution regularizes the singular limit: for fixed \(\Delta_g>0\), \(\eps\downarrow0\) gives a plateau of order \(\Delta_g^{1-d}\), while refinement restores the continuum divergence as \(\Delta_g\downarrow0\).

The same cutoff enters finite-time thermodynamics.  If the friction length has the same leading scaling with constant \(C_\zeta\), the quadratic fixed-duration bound obeys
\begin{equation}
    W_{{\rm ex},{\rm op}}^{(2)}
    \gtrsim
    \frac{C_\zeta^2}{\tau}\,
    \max\{\eps,\Delta_g\}^{-2(d-1)}.
    \label{eq:operational_fixed_time_work}
\end{equation}
Conversely, for a prescribed quadratic excess-work budget \(W_0\) and duration \(\tau\), the achievable operational resolution satisfies the scaling bound
\begin{equation}
    \eps_{\rm ach}
    \gtrsim
    \max\!\left\{
    \Delta_g,
    \left(\frac{C_\zeta^2}{W_0\tau}\right)^{1/[2(d-1)]}
    \right\},
    \label{eq:achievable_resolution_bound}
\end{equation}
with constants depending on the chosen metric convention and window.  This formula separates two regimes: a measurement-limited regime controlled by \(\Delta_g\), and a thermodynamic-budget-limited regime controlled by \(W_0\tau\).

A complementary discrete-state view is obtained by replacing the continuum window by a weighted graph
\begin{equation}
    G_{\Delta}=(V_{\Delta},E_{\Delta}),
    \label{eq:weighted_graph_def}
\end{equation}
whose vertices are cells or cell centers and whose local edge weights are
\begin{equation}
    w_{ij}\simeq
    \left[\Delta\lambda_{ij}^{a}\,
    \zeta_{ab}(\lambda_{ij}^{\rm mid})
    \Delta\lambda_{ij}^{b}\right]^{1/2}.
    \label{eq:graph_edge_weight}
\end{equation}
\HL_REF_B{Under the stated partition assumptions, \(|V_\Delta|=\Theta(\Delta_g^{-d})\) and local edge weights are \(\Theta(\Delta_g)\). A spanning-tree walk visits every cell with at most \(O(|V_\Delta|)\) local transitions, even when a Hamiltonian path is unavailable. Together with the coverage lower bound this gives the order \(\Delta_g^{1-d}\) in Eq.~\eqref{eq:operational_length_scaling}.} The graph formulation specifies serial exhaustive coverage; it does not solve a traveling-salesman optimization problem.

\HL_REF_AB{The allocation of time over this graph determines the quadratic cost. At fixed total duration and constant friction speed, Eq.~\eqref{eq:operational_fixed_time_work} gives \(W_{\rm ex}^{(2)}\sim\Delta_g^{-2(d-1)}/\tau\). If each moving segment instead receives a fixed duration \(t_{\rm move}\), then}
\begin{equation}
    \HL_REF_B{\tau_{\rm move}\sim |V_\Delta|t_{\rm move}
    =\Theta(\Delta_g^{-d}t_{\rm move})},
    \label{eq:fixed_dwell_total_time}
\end{equation}
\HL_REF_AB{whereas its discrete quadratic action scales as}
\begin{equation}
    \HL_REF_B{W_{\rm move}^{(2)}
    \sim\sum_{(i,j)}\frac{w_{ij}^{2}}{t_{\rm move}}
    =\Theta\!\left(\frac{\Delta_g^{2-d}}{t_{\rm move}}\right)}.
    \label{eq:fixed_dwell_work}
\end{equation}
\HL_REF_AB{In \(d=2\), this action remains of constant order while the time spent moving grows as \(\Delta_g^{-2}\). This is a time-allocation consequence within the quadratic model, conditional on the physical slow-driving regime.} \HL_REF_B{The segment duration \(t_{\rm move}\), called dwell time in the original discrete estimator, is not a holding or measurement time. If each cell additionally requires equilibration and measurement, these durations must be added explicitly, for example as \(N_{\rm cell}(t_{\rm eq}+t_{\rm meas})\). Cell count is not a count of statistically independent samples; the latter also depends on autocorrelation times and the acquisition procedure.}

\subsection{Order of limits for cutoff-dependent windows}
\label{subsec:order_of_limits}

When a critical point, spinodal, or other singular set is excluded by a cutoff \(\Delta>0\), the regular theorem is always a fixed-window statement.  Write the cutoff window as \(\M_\Delta\) and the corresponding smooth positive metric as \(g_\Delta\).  For every fixed \(\Delta\), Theorem~\ref{thm:length_lower_bound} gives the \(\eps\downarrow0\) lower bound on \(\M_\Delta\), and Hilbert/Peano-type local traversals give the matching exponent.  In this sense one writes, for a specified traversal family,
\begin{equation}
    L_\eps(\Delta)\sim C_\Delta\,\eps^{1-d},
    \qquad \eps\downarrow0\quad\text{at fixed }\Delta .
    \label{eq:fixed_delta_resolution_limit}
\end{equation}
The critical-window problem is the subsequent behavior of the fixed-window amplitude \(C_\Delta\) as \(\Delta\downarrow0\).  Equivalently, when the prefactor limit exists for the chosen traversal,
\begin{equation}
    C_\Delta=\lim_{\eps\downarrow0}\eps^{d-1}L_\eps(\Delta),
    \qquad \text{then study }\quad \Delta\downarrow0 .
    \label{eq:sequential_limits}
\end{equation}
\HL_REF_A{For the numerical amplitudes, \(\eps\) is the stated coordinate-cell covering scale.  Their \(C_\Delta\) therefore belongs to that scale convention and traversal family; it is not the minimum intrinsic covering prefactor over all admissible protocols.  In a cutoff-dependent window, the constants comparing coordinate and intrinsic resolution may themselves depend on \(\Delta\).} The analytic order of limits in the critical-prefactor statements is therefore \(\lim_{\Delta\downarrow0}\lim_{\eps\downarrow0}\), not an uncontrolled simultaneous limit.  The finite numerical data necessarily use finite grid orders and finite cutoffs; the reported exponents and highest-order prefactors should be read as finite-window approximations to Eq.~\eqref{eq:fixed_delta_resolution_limit}.  If the grid scale is not asymptotically smaller than the critical-core width, the fitted slopes are effective diagnostics rather than precision asymptotic exponents.

\subsection{Model systems and explicit metrics}
\label{subsec:model_systems}

This subsection specifies the three model classes used in the numerical analysis.  The purpose here is not yet to analyze critical scaling, but to provide explicit state spaces, metrics, and regularity assumptions for the finite-resolution length bound.

\subsubsection{Ideal gas}
\label{subsubsec:ideal_gas}
\HL_REF_A{For a single-component ideal gas at fixed particle number, the per-particle coordinates $u,v>0$ give the regular benchmark}
\begin{equation}
    \HL_REF_A{g^{\rm id}_{ij}}
    =
    \begin{pmatrix}
        c/u^2 & 0 \\
        0 & 1/v^2
    \end{pmatrix}.
    \label{eq:ideal_metric}
\end{equation}
\HL_REF_A{with dimensionless heat capacity $c>0$. It is flat under}
\begin{equation}
    y^1=\sqrt{c}\ln u,
    \qquad
    y^2=\ln v,
    \label{eq:ideal_flat_coords}
\end{equation}
\HL_REF_A{so it checks the geometric exponent without curvature or critical effects. The numerical window is $[1,4]^2$, $c=1.5$. Entropy, volume and coordinate details are given in Appendix~\ref{sm:ideal_metric_details}.}

\subsubsection{\HL_REF_A{van der Waals benchmarks}}
\label{subsubsec:vdw}
\HL_REF_AB{The regular van der Waals Ruppeiner benchmark uses the $(T,v)$ pullback in Eq.~\eqref{eq:vdw_ruppeiner_pullback_analysis}; the response-proxy metrics in Eq.~\eqref{eq:vdw_response_proxy_analysis} are distinct diagnostic definitions. Their common numerical window is $T\in[T_c+\Delta,T_c+\Delta+0.15]$, $v\in[2.5,3.5]$, with $a=b=1$, $c_v=1.5$, $T_c=8/27$ and $\Delta>0$. The entropy Hessian and admissible-domain formulas are retained in Appendix~\ref{sm:vdw_metric_details}. The present calculations use the specified coordinate grid, not an adaptive or volume-uniform mesh.}

\subsubsection{\HL_REF_AB{Landau-response diagnostic coverage geometry}}
\label{subsubsec:ising}

\HL_REF_AB{The Ising-related calculations below use prescribed diagnostic geometries built from a Landau response function.  They do not use the full microscopic Fisher metric of the Curie--Weiss model.  We first make this distinction explicit.}  For \(N\) spins \(\sigma_i=\pm1\), the Curie--Weiss Hamiltonian is
\begin{equation}
    \mathcal{H}_N(\sigma;J,h)
    =-\frac{J}{2N}\left(\sum_{i=1}^N\sigma_i\right)^2
    -h\sum_{i=1}^N\sigma_i,
    \label{eq:curie_weiss_hamiltonian}
\end{equation}
with \(J>0\).  The partition function is
\begin{equation}
    Z_N(\beta,h)
    =\sum_{\{\sigma_i\}}
    \exp\left[
    \frac{\beta J}{2N}\left(\sum_i\sigma_i\right)^2
    +\beta h\sum_i\sigma_i
    \right].
    \label{eq:curie_weiss_partition}
\end{equation}
In natural parameters \(\theta=(\beta J,\beta h)\), the finite-size Fisher metric per spin is
\begin{equation}
    g^{(N)}_{ij}(\theta)
    =\frac{1}{N}\partial_i\partial_j\ln Z_N(\theta).
    \label{eq:ising_fisher_finiteN}
\end{equation}
\HL_REF_B{This is the covariance per spin of the sufficient statistics conjugate to \(\theta\).  Its thermodynamic-limit counterpart cannot be assumed to be a positive-definite two-dimensional metric.}  The thermodynamic-limit pressure has the variational expression
\begin{equation}
    \psi(\theta)
    =\sup_{m\in[-1,1]}
    \left[\frac{\theta^1}{2}m^2+\theta^2m+s_{\rm mix}(m)\right],
    \label{eq:ising_pressure_variational}
\end{equation}
where
\begin{equation}
    s_{\rm mix}(m)
    =-\frac{1+m}{2}\ln\frac{1+m}{2}
     -\frac{1-m}{2}\ln\frac{1-m}{2}.
    \label{eq:mixing_entropy}
\end{equation}
The maximizing magnetization satisfies
\begin{equation}
    m=\tanh(\beta J m+\beta h).
    \label{eq:mean_field_equation}
\end{equation}
\HL_REF_B{On a smooth, unique stable branch, differentiation gives \(\nabla_\theta\psi=(m^2/2,m)\) and}
\begin{equation}
\HL_REF_B{\begin{aligned}
    \nabla_\theta^2\psi
      &=\frac{1}{D}\begin{pmatrix}m^2&m\\m&1\end{pmatrix},\\
    D&=(1-m^2)^{-1}-\theta^1>0.
\end{aligned}}
\label{eq:ising_rank_one_reference}
\end{equation}
\HL_REF_B{This Hessian has rank one even away from criticality.  Excluding the critical region, or making a regular coordinate change to \((T,h)\), does not remove its zero eigenvalue.  We therefore do not apply Theorem~\ref{thm:length_lower_bound} to this limiting microscopic Hessian as a two-dimensional positive-definite metric.  Nor do the computations evaluate the finite-\(N\) covariance in Eq.~\eqref{eq:ising_fisher_finiteN}; there is no finite-size limit hidden in the numerical protocol.}

\HL_REF_AB{Instead, the diagnostic response is defined by the Landau potential}
\begin{equation}
    \begin{aligned}
    \Phi(T,m;h)&=\Phi_0(T)+\frac{a}{2}(T-T_c)m^2\\
               &\quad+\frac{b}{4}m^4-hm,\qquad a,b>0,
    \end{aligned}
    \label{eq:landau_potential}
\end{equation}
with equilibrium condition
\begin{equation}
    h=a(T-T_c)m+bm^3.
    \label{eq:landau_equation_state}
\end{equation}
The magnetic susceptibility in this approximation is
\begin{equation}
    \chi_T=\left(\frac{\partial m}{\partial h}\right)_T
    =\frac{1}{a(T-T_c)+3bm^2}.
    \label{eq:landau_susceptibility}
\end{equation}
\HL_REF_B{We prescribe the static diagnostic tensor}
\begin{equation}
    \HL_REF_B{\begin{aligned}
    g^{\rm st}&=\operatorname{diag}(C/T^2,\beta\chi_T),\\
    C&>0,\qquad\beta=1/T.
    \end{aligned}}
    \label{eq:ising_static_diagnostic}
\end{equation}
\HL_REF_B{The regular component \(C/T^2\) is an assumed positive weight, not a heat capacity derived from the Curie--Weiss Hamiltonian or from \(\Phi_0\).  The cross term is zero by definition, not by a diagonalization of the microscopic covariance tensor in the same coordinates.  The tensor is smooth and positive definite on the compact window}
\begin{equation}
\HL_REF_B{\begin{aligned}
    \M_{\rm Ising}(\Delta)
      &=[T_c+\Delta,T_c+\Delta+t_w]\\
      &\quad\times[-h_0,h_0],\qquad\Delta>0,
\end{aligned}}
    \label{eq:ising_regular_domain}
\end{equation}
\HL_REF_B{where \(t_w=0.25\), \(h_0=0.2\), and \(T_c=a=b=C=1\) in the computations.  Theorem~\ref{thm:length_lower_bound} applies to this prescribed geometry at every fixed \(\Delta\), without identifying it with Eq.~\eqref{eq:ising_fisher_finiteN} or Eq.~\eqref{eq:ising_rank_one_reference}.}

\subsubsection{\HL_REF_AB{Response-time weighting as a diagnostic ansatz}}
\label{subsubsec:ising_friction}

\HL_REF_B{The magnetic normalization can be fixed independently of the choice of the full diagnostic tensor.  For the physical field coupling \(-hM\), with \(m=\langle M\rangle/N\) and \(k_B=1\),}
\begin{equation}
\HL_REF_B{\begin{aligned}
    \chi_T&=\frac{\beta}{N}\operatorname{Var}(M),\\
    g^{\rm F}_{hh}&=\frac{\beta^2}{N}\operatorname{Var}(M)
                  =\beta\chi_T.
\end{aligned}}
    \label{eq:ising_field_normalisation}
\end{equation}
\HL_REF_B{Thus retaining \(\beta\chi_T\) in Eq.~\eqref{eq:ising_static_diagnostic} is consistent with the physical-field Fisher component, although it does not derive the full diagonal tensor.  For a single-mode approximation to the magnetic Kubo correlation at fixed temperature \cite{SivakCrooks2012,Zulkowski2012}, the per-spin friction is}
\begin{equation}
    \HL_REF_B{\zeta_{hh}^{\rm per\,spin}
    \simeq\beta\,\tau_{\rm relax}\,
          \frac{\operatorname{Var}(M)}{N},}
    \label{eq:friction_factorization}
\end{equation}
\begin{equation}
    \HL_REF_B{\begin{aligned}
    \zeta_{hh}^{\rm per\,spin}
      &\simeq\tau_{\rm relax}\chi_T
      =\tau_0\chi_T^{\,1+z\nu},\\
    \tau_{\rm relax}&=\tau_0\chi_T^{\,z\nu}.
    \end{aligned}}
    \label{eq:friction_hh}
\end{equation}
\HL_REF_B{There is no additional factor \(\beta\) in Eq.~\eqref{eq:friction_hh}.  This normalization motivates the response-time weighted diagnostic tensor}
\begin{equation}
\HL_REF_B{\begin{aligned}
    g^{\rm wt}
      &=\operatorname{diag}\!\left(
          \tau_T C/T^2,\tau_0\chi_T^{\,1+z\nu}\right),\\
    \tau_T&>0,\qquad\tau_0>0,
\end{aligned}}
    \label{eq:ising_weighted_diagnostic}
\end{equation}
\HL_REF_AB{with \(\tau_T=\tau_0=1\) and \(z\nu=1\) in the baseline run.  The response-time power is motivated by mean-field Model-A scaling, but Eq.~\eqref{eq:ising_weighted_diagnostic} remains a prescribed diagnostic ansatz.  Its temperature component, zero cross term, and complete relaxation dynamics have not been derived microscopically.  Consequently, lengths computed from this tensor test a geometry with enhanced response; they do not independently predict physical dissipation in a critical Ising system.}

\HL_REF_B{For the two diagonal tensors the generalized-eigenvalue ratios are \(\tau_T\) and \(T\tau_0\chi_T^{z\nu}\).  The continuous ratios have positive finite extrema on each fixed compact window. We estimate them on a \(129\times129\) grid. For the five baseline cutoffs in descending order, the sampled upper values are \(11,21,41,81,161\), while the sampled lower values are \(1,1,1,0.998876,0.986738\) (Fig.~\ref{fig:friction_comparability}).  The sampled values are diagnostics, not certified numerical bounds on the continuous extrema. This is a comparison between two defined diagnostic geometries.  A transfer of the physical dissipation bounds would additionally require a separately justified friction tensor satisfying Eq.~\eqref{eq:zeta_g_domination}; the numerical comparison alone does not establish that hypothesis for microscopic Ising dynamics.}

\subsection{Operational summary of the finite-resolution protocol}
\label{subsec:operational_summary}

The preceding definitions lead to a concrete protocol for both analytic estimates and numerical experiments:
\begin{enumerate}
    \item Choose a compact regular thermodynamic window \(\M\) and a thermodynamic metric \(g\) for coverage.
    \item \HL_REF_AB{For a physical-work interpretation, specify an independently justified friction tensor \(\zeta\); the choice \(\zeta=g\) requires that the coverage metric itself be such a tensor.  The Markov-jump and harmonic-trap examples in Secs.~\ref{subsec:markov_jump_friction} and~\ref{subsec:harmonic_trap_friction} satisfy this requirement.  A diagnostic geometry can instead be studied without assigning its length a physical dissipation meaning.}
    \item Construct an \(\eps\)-dense Hilbert/Peano-type grid traversal \(H_\eps\) in coordinates adapted either to the experimental controls or to the Riemannian volume element \(\dd\mu_g\).
    \item Compute or estimate \(\Length_g[H_\eps]\) and \(\Length_\zeta[H_\eps]\) using Eq.~\eqref{eq:discrete_length} or a geodesic quadrature rule.
    \item Compare the observed scaling with the universal lower bound \(\Length_g[H_\eps]\ge C_g\eps^{1-d}-O(\eps)\), the quadratic finite-time dissipation bound \(W_{\rm ex}^{(2)}\ge \Length_\zeta^2/\tau\), and, when a finite observation floor is present, the operational cutoff law \(\Length_{\rm op}=\Theta(\max\{\eps,\Delta_g\}^{1-d})\).
\end{enumerate}
For regular compact state spaces, the expected leading exponent is fixed by dimension.  Deviations from this behavior are therefore diagnostic of finite-resolution effects, metric singularities, boundary effects, or a nonuniform sampling prescription.  The analysis below uses this regular result as the baseline against which critical-window scaling is measured.

\begin{figure}[t]
    \centering
    \includegraphics[width=0.46\textwidth]{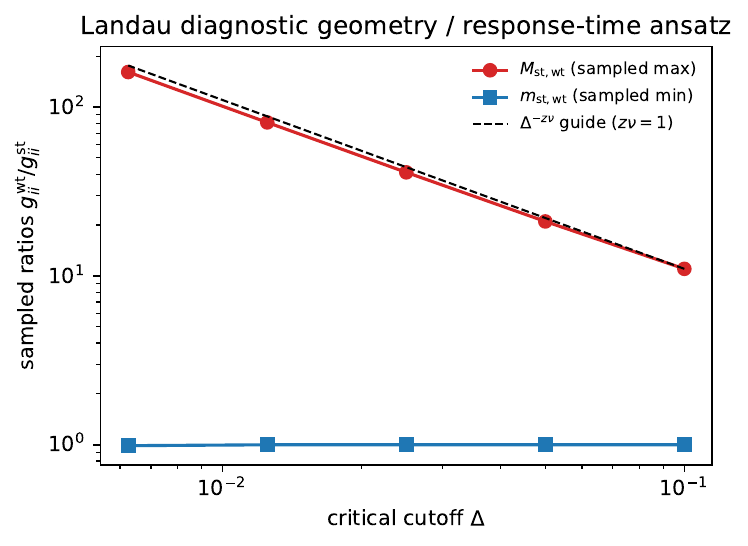}
    \caption{\HL_REF_AB{Comparison of the prescribed static and response-time weighted Landau diagnostic tensors, Eqs.~\eqref{eq:ising_static_diagnostic} and~\eqref{eq:ising_weighted_diagnostic}.  Sampled extrema of \((g^{\rm st})^{-1}g^{\rm wt}\) on a \(129\times129\) grid estimate the lower comparison value \(m_{\rm st,wt}\) (squares) and upper comparison value \(M_{\rm st,wt}\) (circles).  The dashed \(\Delta^{-z\nu}\) line is a scaling guide, not an exact finite-cutoff formula; \(z\nu=1\).  The magnetic ratio is \(T\chi_T\) after the normalization correction.  Compactness and positive definiteness guarantee the existence of global comparison constants; the sampled values are not certified bounds and do not constitute a microscopic Ising friction measurement or a physical-work validation.}}
    \label{fig:friction_comparability}
\end{figure}

\section{Analysis}
\label{sec:analysis}

\subsection{Numerical protocol and estimators}
\label{subsec:numerical_protocol}

\HL_REF_AB{We first benchmark the implementation against the known geometric resolution exponent in two-dimensional windows; these length fits are not independent tests of physical dissipation.}  The baseline numerical curves are Hilbert center-to-center traversals of rectangular coordinate windows.  An order-\(n\) traversal contains \(2^n\times 2^n\) cells and has coordinate covering scale \(\eps\propto 2^{-n}\).  The reported \(\eps\) is this coordinate covering radius.  More precisely, if \(\eps_g\) denotes the intrinsic covering radius, Eq.~\eqref{eq:coord_intrinsic_distance} implies
\begin{equation}
    a_\Delta \eps \le \eps_g \le b_\Delta \eps
    \label{eq:numerical_epsilon_comparison}
\end{equation}
for each fixed regular numerical window, with positive constants determined by the metric and the chosen coordinates.  Hence fitting in terms of the coordinate scale \(\eps\) gives the same exponent as fitting in terms of \(\eps_g\).  \HL_REF_A{The reported \(C_\Delta\) belongs to the prescribed coordinate grid, ordering, and connector estimator.  It is not the minimum intrinsic covering prefactor of the manifold.  The constants \(a_\Delta,b_\Delta\) may depend on the cutoff, so the fixed-window equivalence of resolution exponents does not make the cutoff dependence of \(C_\Delta\) coordinate independent.}

For an ordered list of grid centers \(x_0,x_1,\ldots,x_N\), the thermodynamic length was evaluated by midpoint quadrature,
\begin{equation}
    L_\eps^{\mathcal T}
    =\sum_{n=0}^{N-1}
    \left[
    \Delta x_n^i\,
    g_{ij}\!\left(\frac{x_n+x_{n+1}}{2}\right)
    \Delta x_n^j
    \right]^{1/2},
    \label{eq:analysis_discrete_length}
\end{equation}
where \(\Delta x_n=x_{n+1}-x_n\) and \(\mathcal T\) denotes the traversal.  For Hilbert traversal, \(\mathcal T=\mathrm{H}\), consecutive cells share an edge.  For Morton/Z-order traversal, \(\mathcal T=\mathrm{M}\), the same cells are ordered by the key in Eq.~\eqref{eq:morton_key}; consecutive cells may therefore be separated by nonlocal jumps.  Equation~\eqref{eq:analysis_discrete_length} uses the full quadratic form for both traversals, so diagonal and long Morton jumps are not reduced to axis-aligned steps.  \HL_REF_A{For Morton this is the specified midpoint coordinate-segment estimator for an ordering/connector benchmark.}  A laboratory protocol that realizes the same ordering with different connectors would have to recompute the corresponding length using those connectors.  Component decompositions for Hilbert are obtained from the axis-aligned segments.  For Morton they are used only as projected diagnostic contributions, while the total length is always the full metric length in Eq.~\eqref{eq:analysis_discrete_length}.

For each model, traversal, and, when applicable, critical cutoff \(\Delta\), the resolution dependence was fitted as
\begin{equation}
    L_\eps^{\mathcal T}(\Delta)=C_{\Delta,\mathcal T}^{\rm fit}\,\eps^{-p_{\mathcal T}}.
    \label{eq:length_fit_analysis}
\end{equation}
The standard-window fits used the tail orders \(n=5,6,7,8\).  In accordance with Sec.~\ref{subsec:order_of_limits}, these fits approximate the fixed-\(\Delta\) resolution limit first; the cutoff dependence is then inferred from the resulting amplitudes.  Critical-prefactor plots use the highest available resolution estimator
\begin{equation}
    C_{\Delta,\mathcal T}^{\rm hi}=L_{\eps_{\min}}^{\mathcal T}(\Delta)\,\eps_{\min},
    \label{eq:highest_order_prefactor}
\end{equation}
for the two-dimensional runs.  \HL_REF_B{This is a finite-grid approximation to the fixed-window amplitude in Eq.~\eqref{eq:fixed_delta_resolution_limit}, distinct from the fitted \(C_{\Delta,\mathcal T}^{\rm fit}\) when the fitted \(p\) differs from one.  Passing the critical-core resolution flag is not a convergence certificate; order-by-order prefactors are examined separately in Table~\ref{tab:critical_grid_convergence}.}  \HL_REF_A{The ideal-gas subdivision check splits every Hilbert and Morton jump into four equal coordinate subsegments before midpoint quadrature.  It measures sensitivity to this quadrature choice while retaining the same ordering and coordinate connector path.}

\HL_REF_A{A source-code-independent specification of the numerical protocol is given in Appendix~\ref{app:numerical_details}.}  Briefly, all \HL_REF_B{geometric benchmark} runs are deterministic center-to-center grid traversals with no random sampling.  The baseline Hilbert data use orders $n=4,5,6,7,8$ and the tail fit $n=5,6,7,8$.  The $q=3$ and $q=4$ van der Waals response-proxy critical-prefactor runs use orders $n=6,7,8,9$, with an additional small-$\Delta$ $q=4$ check using $n=8,9,10$.  The Morton/Z-order comparison uses the same baseline orders and, for the nonintegrable $q=4$ stress test, orders $n=6,7,8,9$.  \HL_REF_B{The quoted slope standard errors describe the log--log ordinary-least-squares regression over these orders.  They do not estimate discretization error, and they are not uncertainties from statistically independent realizations.}

The numerical models are as follows.  The ideal gas uses the Ruppeiner metric in Eq.~\eqref{eq:ideal_metric}.  \HL_REF_AB{The Landau-response examples use Eq.~\eqref{eq:landau_equation_state} to define two diagonal diagnostic geometries, as specified in Sec.~\ref{subsubsec:ising}.  Their regular temperature component and zero cross term are prescribed, not derived from a microscopic Ising Fisher tensor.  The static diagnostic metric has}
\begin{equation}
    \HL_REF_B{g^{\rm st}_{tt}=\frac{C}{T^2}},
    \qquad
    \HL_REF_B{g^{\rm st}_{hh}=\beta\chi}
\end{equation}
\HL_REF_B{and the relaxation-weighted diagnostic metric instead has}
\begin{equation}
    \HL_REF_B{g^{\rm wt}_{hh}=\chi^{1+z\nu}}
    \label{eq:ising_friction_metric_analysis}
\end{equation}
\HL_REF_B{with \(z\nu=1\) and \(\tau_T=\tau_0=1\) in the baseline run.  The magnetic component follows the per-spin physical-field convention \(\zeta_{hh}\simeq\tau_{\rm relax}\chi\); it contains no additional \(\beta\).  This normalization does not supply a microscopic dynamics for the full diagnostic tensor.}  The van der Waals critical windows are expressed in \((T,v)\) coordinates.  Defining
\begin{equation}
    s(T,v)=\frac{1}{(v-b)^2}-\frac{2a}{T v^3},
    \label{eq:vdw_stability_factor}
\end{equation}
with \(T_c=8a/(27b)\) and \(v_c=3b\), the Ruppeiner-pullback benchmark uses
\begin{equation}
    g_{TT}=\frac{c_v}{T^2},\qquad g_{vv}=s(T,v),
    \label{eq:vdw_ruppeiner_pullback_analysis}
\end{equation}
whereas the response-proxy family uses
\begin{equation}
    g_{TT}=\frac{c_v}{T^2},\qquad g_{vv}=s(T,v)^{-q}.
    \label{eq:vdw_response_proxy_analysis}
\end{equation}
The family in Eq.~(\ref{eq:vdw_response_proxy_analysis}) is a diagnostic response proxy, not a microscopic finite-time friction tensor.  It fixes the regular \(T\)-direction component and tunes only the singular \(v\)-direction length density through \(q\).  Thus \(q\) is not a dynamic critical exponent, transport coefficient, or universal property of the van der Waals fluid.  Its purpose is to vary the integrability of \(\sqrt{g_{vv}}\) on a cutoff window in a controlled Riemannian model.  If a microscopic friction tensor \(\zeta\), for example from generalized-force time correlations, is uniformly comparable to this proxy on the same window, then the dissipation bounds of Sec.~\ref{subsec:dissipation} apply with \(g\) replaced by \(\zeta\).  Otherwise Eq.~(\ref{eq:vdw_response_proxy_analysis}) is only a diagnostic metric for isolating how a singular directional length density changes the space-filling prefactor.

The baseline proxy has \(q=1\).  The \(q=3\) and \(q=4\) runs are controlled geometric stress tests for critical integrability; no universal dynamic or dissipative exponent is inferred from them.  Hilbert/Morton comparisons were performed for the baseline models and the nonintegrable \(q=4\) stress test.

\subsection{Operational resolution cutoff and resource allocation}
\label{subsec:operational_resolution_results}

\HL_REF_AB{We first illustrate the resource laws using the harmonic-trap friction metric of Sec.~\ref{subsec:harmonic_trap_friction}. We distinguish two calculations: a geometric evaluation of the quadratic action, followed by an independent evaluation of work from driven dynamics.} The window is \(a\in[0,1]\), \(k\in[1,4]\), with \(\beta=\mu=1\).

\subsubsection{\HL_REF_AB{Quadratic-action resource estimates}}

A serpentine raster with \(2^n\times2^n\) cells, \(n=3,\ldots,9\), is evaluated by the midpoint length estimator of Eq.~\eqref{eq:analysis_discrete_length}. The coordinate covering scale is the half diagonal of a grid cell. For fixed total time \(\tau=1\), the quadratic estimate is
\begin{equation}
    W_{\rm fixed\;time}^{(2)}=\frac{L_\zeta^2}{\tau}.
    \label{eq:analysis_fixed_time_work}
\end{equation}
\HL_REF_B{For the alternative allocation, every moving segment receives \(t_{\rm move}=1\), giving}
\begin{equation}
    \HL_REF_B{W_{\rm move}^{(2)}=\sum_i\frac{\ell_i^2}{t_{\rm move}},
    \qquad \tau_{\rm move}=N_{\rm seg}t_{\rm move}},
    \label{eq:analysis_fixed_dwell_estimators}
\end{equation}
where \(\ell_i\) is the segment friction length.

\HL_REF_AB{Figure~\ref{fig:operational_resolution} combines the cutoff law and these quadratic estimates. Panels (a) and (c) are visualizations of the prescribed effective-resolution formula and its budget boundary, rather than independent measurements of a coarse-graining plateau. Panel (b) uses the computed raster lengths. The fitted exponents are}
\begin{equation}
    \HL_REF_AB{L_\zeta\sim\eps^{-1.005},
    \qquad W_{\rm fixed\;time}^{(2)}\sim\eps^{-2.011}},
    \label{eq:operational_fit_exponents_main}
\end{equation}
\HL_REF_AB{and}
\begin{equation}
    \HL_REF_AB{\tau_{\rm move}\sim\eps^{-2.000},
    \qquad W_{\rm move}^{(2)}\sim\eps^{-0.007}}.
    \label{eq:operational_dwell_fit_exponents_main}
\end{equation}
\HL_REF_B{The factor of two between the first two fitted exponents follows algebraically from Eq.~\eqref{eq:analysis_fixed_time_work}; it is not independent dynamical evidence for physical excess work.} \HL_REF_AB{These estimates show how the quadratic resource cost is allocated between work and time. Fixed-time refinement need not remain within the physical slow-driving regime.}

\begin{figure*}[t]
    \centering
    \includegraphics[width=0.32\textwidth]{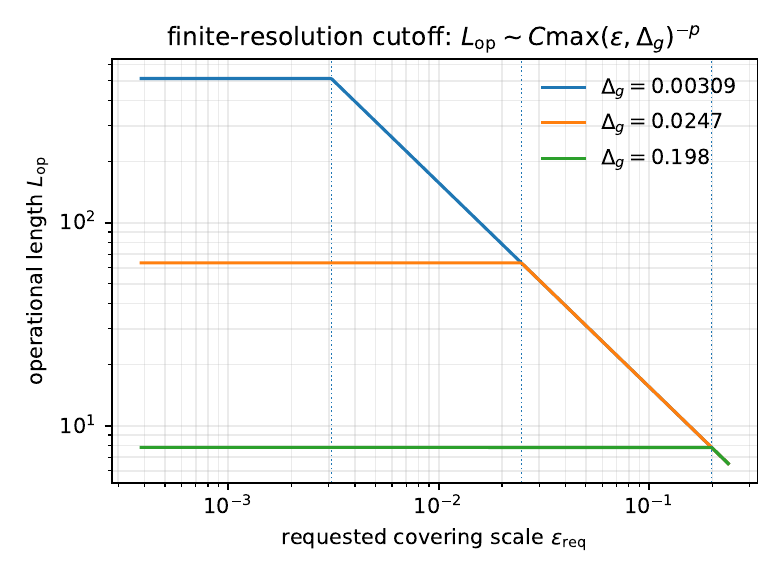}\hfill
    \includegraphics[width=0.32\textwidth]{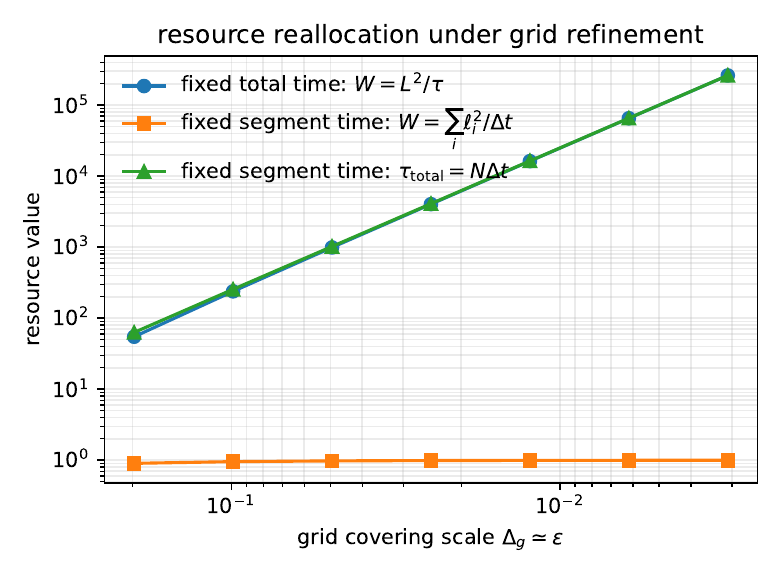}\hfill
    \includegraphics[width=0.32\textwidth]{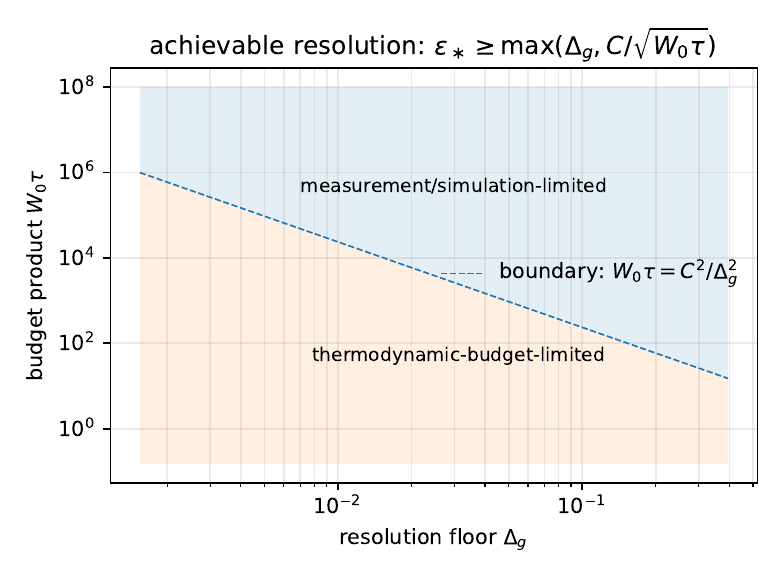}
    \caption{\HL_REF_AB{Quadratic-action resource estimates in the overdamped harmonic-trap metric. (a) Prescribed cutoff form \(L_{\rm op}\sim C\max\{\eps,\Delta_g\}^{-1}\). (b) Computed raster lengths give \(W_{\rm fixed\;time}^{(2)}=L_\zeta^2/\tau\sim\eps^{-2}\); fixed moving-segment duration instead gives an approximately constant discrete action in \(d=2\), with \(\tau_{\rm move}\sim\eps^{-2}\). (c) Resolution--budget boundary implied by Eq.~\eqref{eq:achievable_resolution_bound}. Panels (a,c) display theoretical formulas. The fixed-time quadratic divergence is not a prediction for exact work after slow driving fails.}}
    \label{fig:operational_resolution}
\end{figure*}

\begin{table}[t]
\caption{\HL_REF_AB{Quadratic-action harmonic-trap exponents. Fits use orders \(n=5,\ldots,9\), with \(\tau=1\) or \(t_{\rm move}=1\). The doubled work exponent is a consequence of its definition, not an independent dynamical validation.}}
\label{tab:operational_harmonic_fits}
\begin{ruledtabular}
\begin{tabular}{lc}
Observable & Fitted exponent \\
\hline
\(L_\zeta\) & 1.0054 \\
\(W_{\rm fixed\;time}^{(2)}=L_\zeta^2/\tau\) & 2.0108 \\
\HL_REF_B{\(\tau_{\rm move}=N_{\rm seg}t_{\rm move}\)} & 2.0003 \\
\HL_REF_B{\(W_{\rm move}^{(2)}=\sum_i\ell_i^2/t_{\rm move}\)} & 0.0070
\end{tabular}
\end{ruledtabular}
\end{table}

\subsubsection{\HL_REF_B{Independent driven-dynamics validation}}
\label{subsubsec:trap_dynamic_validation}

\HL_REF_B{For an initially equilibrated Gaussian distribution, the Langevin model has exact closed equations for the mean \(m\), variance \(v\), and mean work \(W\):
\begin{align}
 \dot m&=-\mu k(m-a),\qquad
 \dot v=-2\mu kv+2\mu/\beta,\label{eq:trap_moment_dynamics}\\
 \dot W&=-k(m-a)\dot a+
       \tfrac12\bigl[(m-a)^2+v\bigr]\dot k.
 \label{eq:trap_dynamic_work}
\end{align}
The initial conditions are \(m(0)=a(0)\) and \(v(0)=[\beta k(0)]^{-1}\). We evaluate the dynamic excess work as
\begin{equation}
 W_{\rm ex}^{\rm dyn}=W-\Delta F,
 \qquad \Delta F=\frac{1}{2\beta}\log\frac{k_f}{k_i}.
 \label{eq:trap_dynamic_excess_work}
\end{equation}
This calculation does not impose the quadratic action.}

\HL_REF_B{The controls follow the same cell-center raster, now at orders \(n=3,4,5\). Horizontal moves have linear \(a(t)\) at fixed \(k\), and vertical connectors have linear \(k(t)^{-1/2}\) at fixed \(a\). The exact segment lengths are
\begin{equation}
 \ell_a=\frac{|a_1-a_0|}{\sqrt\mu},\qquad
 \ell_k=\frac{|k_1^{-1/2}-k_0^{-1/2}|}{\sqrt{\beta\mu}}.
 \label{eq:trap_exact_segment_lengths}
\end{equation}
Assigning \(\Delta t_i=\tau\ell_i/L_{\zeta,\rm exact}\) makes friction speed constant within and between segments. The corresponding continuous-metric quadratic action is exactly \(L_{\zeta,\rm exact}^2/\tau\). We keep this exact length separate from the midpoint estimate used in Table~\ref{tab:operational_harmonic_fits}.}

\HL_REF_AB{Figure~\ref{fig:trap_dynamic_validation} compares the independent dynamic work with this prediction for 24 protocols: four fixed durations \(\tau=1,10,100,1000\) at each of three resolutions, and four budget-family durations
\begin{equation}
 \tau_n=\alpha\frac{L_{\zeta,\rm exact}(n)^2}{W_0},
 \qquad W_0=0.05,\quad\alpha=1,4,16,64.
 \label{eq:trap_budget_family}
\end{equation}
At fixed resolution, the ratio \(W_{\rm ex}^{\rm dyn}/(L_{\zeta,\rm exact}^2/\tau)\) approaches unity as duration increases. At fixed short duration, refinement instead produces large departures from the quadratic prediction. Thus the exact dynamics supports the slow-driving interpretation and also displays its finite-time limitation.}

\HL_REF_AB{In the budget family, the dynamic work stays near the quadratic value \(W_0/\alpha\) as resolution is refined while duration grows. At \(\alpha=64\), the maximum relative deviation from that quadratic value across the three resolutions is \(7.06\times10^{-5}\). We additionally inspect segment durations relative to relaxation and sampled mean/variance lags; the duration law alone is not used as a slow-driving guarantee.} \HL_REF_B{The controls are continuous at row turns, but their velocities can change there. The numerical timing, solver tolerances, and local diagnostics are specified in Appendix~\ref{sm:trap_dynamic_details}. No holding, equilibration, or measurement intervals have been added to these protocols.}

\begin{figure*}[t]
 \centering
 \includegraphics[width=0.95\textwidth]{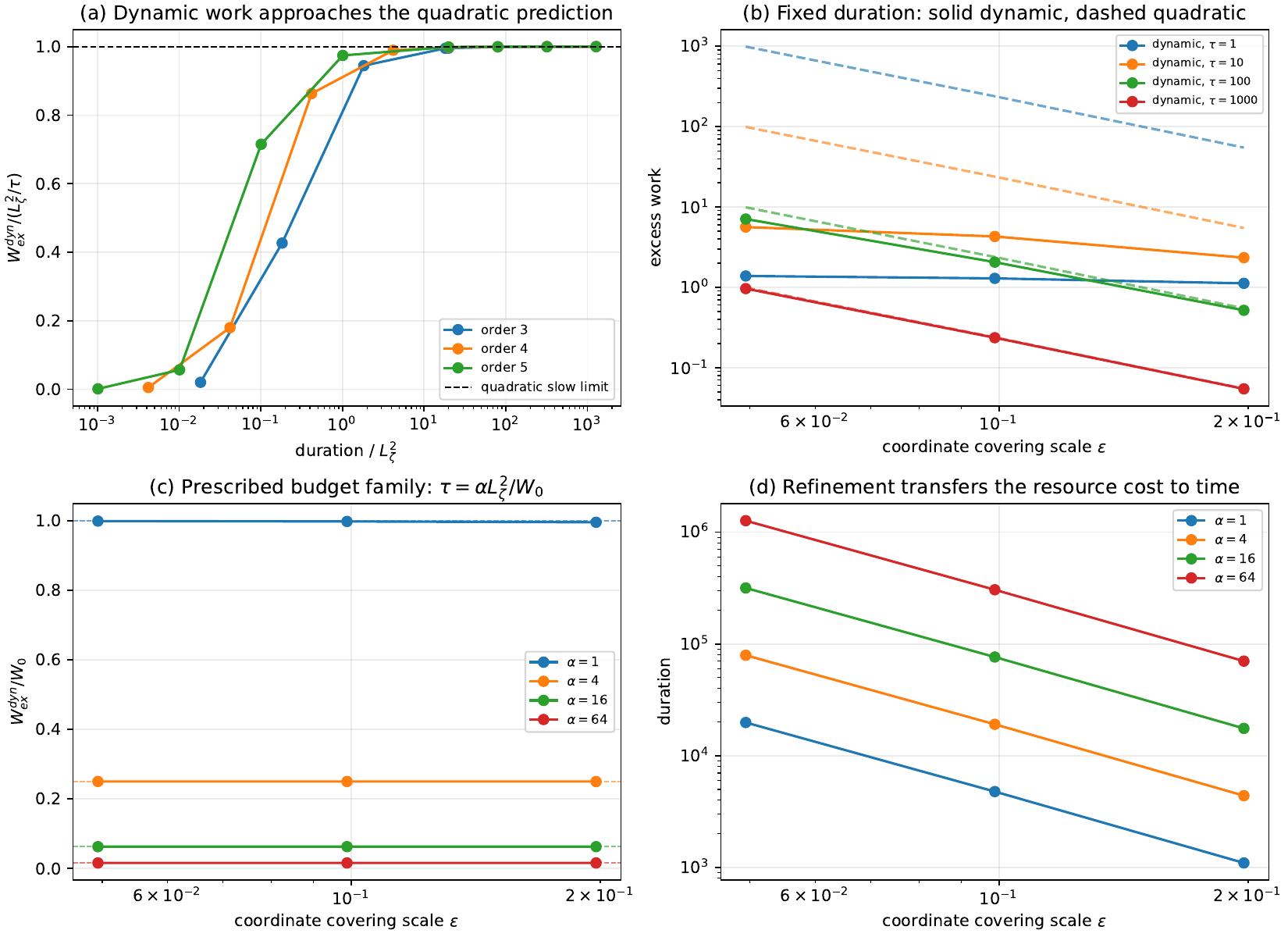}
 \caption{\HL_REF_AB{Independent harmonic-trap dynamics and the scope of the quadratic action. (a) Dynamic excess work divided by \(L_{\zeta,\rm exact}^2/\tau\), for three raster resolutions and increasing durations. (b) At fixed duration, dynamic excess work (solid) can differ strongly from the quadratic prediction (dashed). (c) Prescribed budget-family scans maintain nearly constant work under refinement, close to \(W_0/\alpha\). (d) Their durations increase with refinement. Agreement is checked in the slow regime; no universal inequality \(W_{\rm ex}^{\rm dyn}\ge L_\zeta^2/\tau\) is assumed.}}
 \label{fig:trap_dynamic_validation}
\end{figure*}

\subsection{Resolution scaling in regular and critical windows}
\label{subsec:resolution_scaling_results}

\HL_REF_AB{Figure~\ref{fig:resolution_scaling} benchmarks the finite-grid estimator against the established two-dimensional covering exponent.}  After normalizing by the fitted prefactor, all representative curves collapse onto the reference \(\eps^{-1}\) law over the available orders.  The exponent extracted from Eq.~\eqref{eq:length_fit_analysis} remains close to unity for all baseline models, including the critical-window cases.  \HL_REF_AB{This agreement checks the implementation on the prescribed fixed windows; it does not constitute a new geometric exponent or an independent validation of physical dissipation.}

\begin{figure*}[t]
    \centering
    \includegraphics[width=0.95\textwidth]{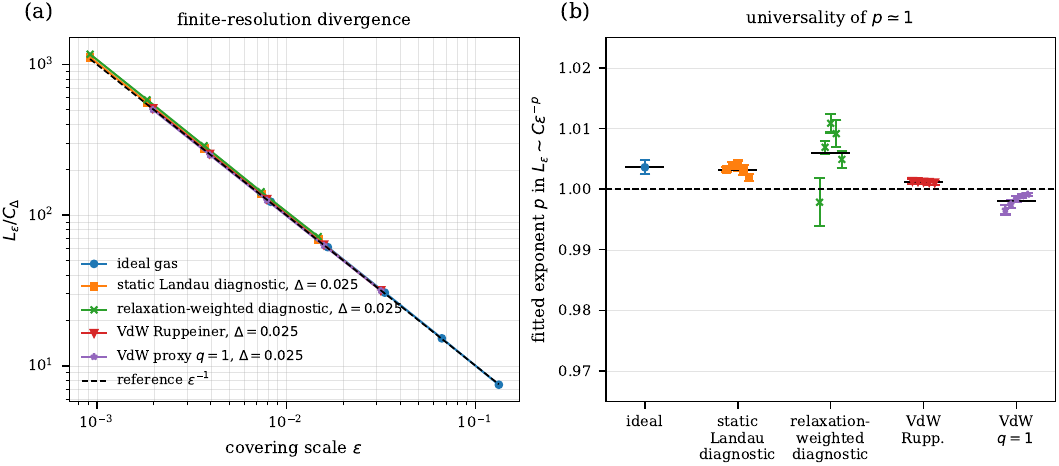}
    \caption{\HL_REF_AB{Finite-resolution length benchmarks for the prescribed two-dimensional geometries.}  (a) Representative normalized lengths \(L_\eps/C_\Delta\) follow the reference \(\eps^{-1}\) law.  The critical-window examples are shown at \(\Delta=0.025\).  (b) Fitted exponents \(p\) in \(L_\eps\sim C\eps^{-p}\) for the baseline models.  Each point corresponds to a cutoff window when a critical cutoff is present; horizontal bars show the modelwise mean.  \HL_REF_B{The two Landau-response series are diagnostic metrics; vertical error bars are regression standard errors.}}
    \label{fig:resolution_scaling}
\end{figure*}

\HL_REF_B{Table~\ref{tab:baseline_exponents} reports the recomputed exponents.  The static Landau diagnostic retains the original normalization, whereas the relaxation-weighted series uses the corrected magnetic component in Eq.~\eqref{eq:ising_friction_metric_analysis}.  Its finite-grid exponent range is now \(0.9978\)--\(1.0109\), and its highest-order prefactor dynamic range is \(1.10\).  The ideal-gas and van der Waals lengths are unchanged.}  \HL_REF_AB{The small fitted departures from unity are finite-grid regression diagnostics; the fixed-window geometric result determines the limiting exponent without requiring an asymptotic critical interpretation of these fits.}

\begin{table*}[t]
\small\setlength{\tabcolsep}{3pt}
\caption{Baseline two-dimensional finite-resolution exponents.  The dynamic range is the ratio of the largest to smallest highest-order prefactor \(C_\Delta^{\rm hi}=L_{\eps_{\min}}\eps_{\min}\) over the listed critical cutoffs.  The ideal gas has no critical cutoff.  \HL_REF_A{These are amplitudes for the specified coordinate grid and traversal family.}}
\label{tab:baseline_exponents}
\begin{ruledtabular}
\begin{tabular}{p{0.34\textwidth}ccc}
Model & Critical cutoffs & Range of fitted \(p\) & Dynamic range of \(C_\Delta^{\rm hi}\) \\
\hline
Ideal gas & none & 1.0036 & -- \\
\HL_REF_B{Landau static diagnostic} & \(0.1\) to \(0.00625\) & 1.0019--1.0039 & 1.07 \\
\HL_REF_B{Landau relaxation-weighted diagnostic, \(z\nu=1\)} & \(0.1\) to \(0.00625\) & \HL_REF_B{0.9978--1.0109} & \HL_REF_B{1.10} \\
VdW Ruppeiner pullback & \(0.1\) to \(0.00625\) & 1.0011--1.0013 & 1.03 \\
VdW response proxy, \(q=1\) & \(0.1\) to \(0.00625\) & 0.9966--0.9991 & 1.46
\end{tabular}
\end{ruledtabular}
\end{table*}

\HL_REF_A{For these prescribed geometries and coordinate protocols, the cutoff dependence of the baseline prefactor is modest.  A locally enhanced component therefore need not produce a divergent directional integral for this family.}  \HL_REF_A{Appendix~\ref{app:standard_diagnostics} collects the baseline cutoff diagnostics, metric-area estimates, and local slopes.}  The two smallest \HL_REF_B{Landau} cutoffs are flagged by the critical-core resolution diagnostic; they are therefore useful for visualizing finite-resolution drift but should not be used to assign an asymptotic critical exponent.  \HL_REF_B{Separate order-convergence data in Table~\ref{tab:critical_grid_convergence} distinguish discretization effects from finite-cutoff corrections.}

\subsection{Critical prefactors and integrability of directional length density}
\label{subsec:critical_integrability_results}

The separation between the resolution exponent and the critical prefactor can be made explicit with the response-proxy family in Eq.~\eqref{eq:vdw_response_proxy_analysis}.  Near the van der Waals critical point,
\begin{equation}
    s(T,v)\simeq A\tau+B u^2,
    \qquad
    \tau=T-T_c,
    \qquad
    u=v-v_c,
    \label{eq:vdw_local_stability_scaling}
\end{equation}
with positive constants \(A\) and \(B\).
\HL_REF_A{The amplitude considered here is defined by the coordinate covering scale and the specified axis-aligned traversal family in Eq.~\eqref{eq:sequential_limits}.  It is not the minimum intrinsic covering amplitude of the manifold.  Although \(a_\Delta\eps\le\eps_g\le b_\Delta\eps\) on every fixed window, the comparison factors can depend on \(\Delta\).  Proposition~\ref{prop:critical_prefactor} consequently classifies this protocol-dependent amplitude, rather than a coordinate-independent critical dissipation cost.}

\begin{proposition}[Critical-window prefactor for a singular directional metric]
\label{prop:critical_prefactor}
Let \(W_\Delta=[\Delta,\tau_0]\times[-u_0,u_0]\), with \(0<\Delta<\tau_0\).  Suppose that, in coordinates \((\tau,u)\), a diagonal metric satisfies
\begin{equation}
\begin{aligned}
    0<c_T &\le g_{\tau\tau}\le C_T<\infty,\\
    c_v(A\tau+B u^2)^{-q} &\le g_{uu}
    \le C_v(A\tau+B u^2)^{-q},
\end{aligned}
    \label{eq:critical_metric_comparison}
\end{equation}
with positive constants independent of \(\Delta\).  \HL_REF_A{Consider a specified family of coordinate grids with uniformly controlled aspect ratios and coordinate covering scale, and a Hilbert-type axis-aligned traversal whose local \(u\)-step density is bounded above and below independently of \(\Delta\).}  \HL_REF_A{Assume that the fixed-window prefactor limit in Eq.~\eqref{eq:sequential_limits} exists for this selected family.}  The singular \(u\)-direction contribution \(C_{\Delta,u}\) to the two-dimensional prefactor in \(L_\eps(\Delta)\sim C_\Delta\eps^{-1}\) is comparable to
\begin{equation}
    I_v(\Delta)
    =
    \int_\Delta^{\tau_0}\!\dd\tau
    \int_{-u_0}^{u_0}\!\dd u\,
    \left(A\tau+B u^2\right)^{-q/2}.
    \label{eq:vdw_directional_integral}
\end{equation}
Consequently,
\begin{equation}
    \HL_REF_A{C_{\Delta,u}\asymp I_v(\Delta)
    \asymp
    \begin{cases}
    1, & q<3,\\[2pt]
    \log(\tau_0/\Delta), & q=3,\\[2pt]
    \Delta^{-(q-3)/2}, & q>3,
    \end{cases}}
    \label{eq:vdw_integrability_classification}
\end{equation}
\HL_REF_A{Here \(I_v\) tends to a finite positive limit for \(q<3\); \(\asymp\) denotes comparison by positive constants independent of sufficiently small \(\Delta\), not equality of leading coefficients.}
\end{proposition}

\begin{proof}
For an axis-aligned Hilbert-type grid traversal at mesh size \(\delta\), the total contribution of \(u\)-oriented steps is a Riemann sum for \(\int_{W_\Delta}\sqrt{g_{uu}}\,\dd\tau\dd u\), multiplied by \(\delta^{-1}\) and by step-density and aspect-ratio constants.  \HL_REF_A{For this coordinate-grid family, the constants relating the coordinate covering scale to \(\delta\) are controlled as \(\Delta\downarrow0\).}  The prefactor \(C_{\Delta,u}=\eps L_{\eps,u}\) has the same \(\Delta\)-dependence as the integral in Eq.~\eqref{eq:vdw_directional_integral}.  The regular \(\tau\)-direction contribution is bounded because \(g_{\tau\tau}\) is uniformly bounded on a finite window.

It remains to classify \(I_v(\Delta)\).  Setting \(u=\sqrt{\tau}\,w\) gives an inner factor \(\tau^{(1-q)/2}\int_{-u_0/\sqrt{\tau}}^{u_0/\sqrt{\tau}}(A+B w^2)^{-q/2}\dd w\).  \HL_REF_A{For \(q>1\), the inner rescaled integral tends to a finite constant and the outer density is proportional to \(\tau^{(1-q)/2}\); at \(q=1\) the inner logarithm is still integrable, and for \(q<1\) the original inner integral stays bounded.  The outer integral is therefore finite for \(q<3\), logarithmic at \(q=3\), and power divergent for \(q>3\).}  Integration over \(\tau\in[\Delta,\tau_0]\) yields Eq.~\eqref{eq:vdw_integrability_classification}.
\end{proof}

\HL_REF_A{The baseline \(q=1\) diagnostic is integrable, \(q=3\) is marginal, and \(q=4\) has the predicted singular-direction power \(1/2\), under the stated traversal hypotheses.  A passive change of coordinates, with the same curve and the transformed tensor, leaves its length unchanged.  Regenerating a uniform Hilbert grid in new coordinates changes the physical sampling and connector family and can change \(C_\Delta\).  The present calculations retain the original coordinate-grid family; optimization over intrinsic or volume-adapted protocols remains a separate problem.}

\HL_REF_AB{The finite-path diagnostics in Fig.~\ref{fig:critical_integrability} remain useful benchmarks, but their effective cutoff exponents are not themselves an asymptotic test.  We therefore compare the path amplitudes with an independently evaluated integral of the actual diagnostic metric,}
\HL_REF_A{\begin{equation}
 J_v^{\rm sl}(\Delta)=\int_\Delta^{\Delta+0.15}\!\dd\tau
 \int_{-0.5}^{0.5}\!\dd u\,
 s(T_c+\tau,v_c+u)^{-q/2}.
 \label{eq:exact_vdw_directional_integral}
\end{equation}}
\HL_REF_A{The superscript ``sl'' denotes the sliding upper limit used in the path calculations.  Equation~\eqref{eq:vdw_directional_integral}, in contrast, has a fixed upper limit \(\tau_0\).  We evaluate both versions of the actual integral and of its Taylor-kernel approximation, using \(\tau_0=0.15\), \(A=27/32\), and \(B=1/48\) for \(a=b=1\).  Their leading singular classes agree, while their finite-cutoff corrections need not.  The independent quadrature does not reuse the path grid or the original grid-based directional-integral estimator.}

\HL_REF_A{Figure~\ref{fig:path_integral_validation} compares the two stages: discrete path to actual-metric integral, and actual-metric integral to local asymptotics.  For the \(q=4\), order-10 path, \(C_{\Delta,v}^{\rm hi}/J_v^{\rm sl}\) is \(1.68594\), \(1.69009\), and \(1.71249\) at \(\Delta=0.0125\), \(0.003125\), and \(0.00078125\), respectively.  This comparatively modest variation supports the directional-integral interpretation, while the rise at the smallest cutoff exposes residual path discretization.  Over the smallest path-cutoff interval, the local exponent of \(C_\Delta^{\rm hi}\) is \(0.62437\), whereas that of \(J_v^{\rm sl}\) is \(0.61154\).  Thus finite-cutoff corrections remain substantial when the actual density is integrated independently of the path.}

\begin{figure*}[t]
    \centering
    \includegraphics[width=0.95\textwidth]{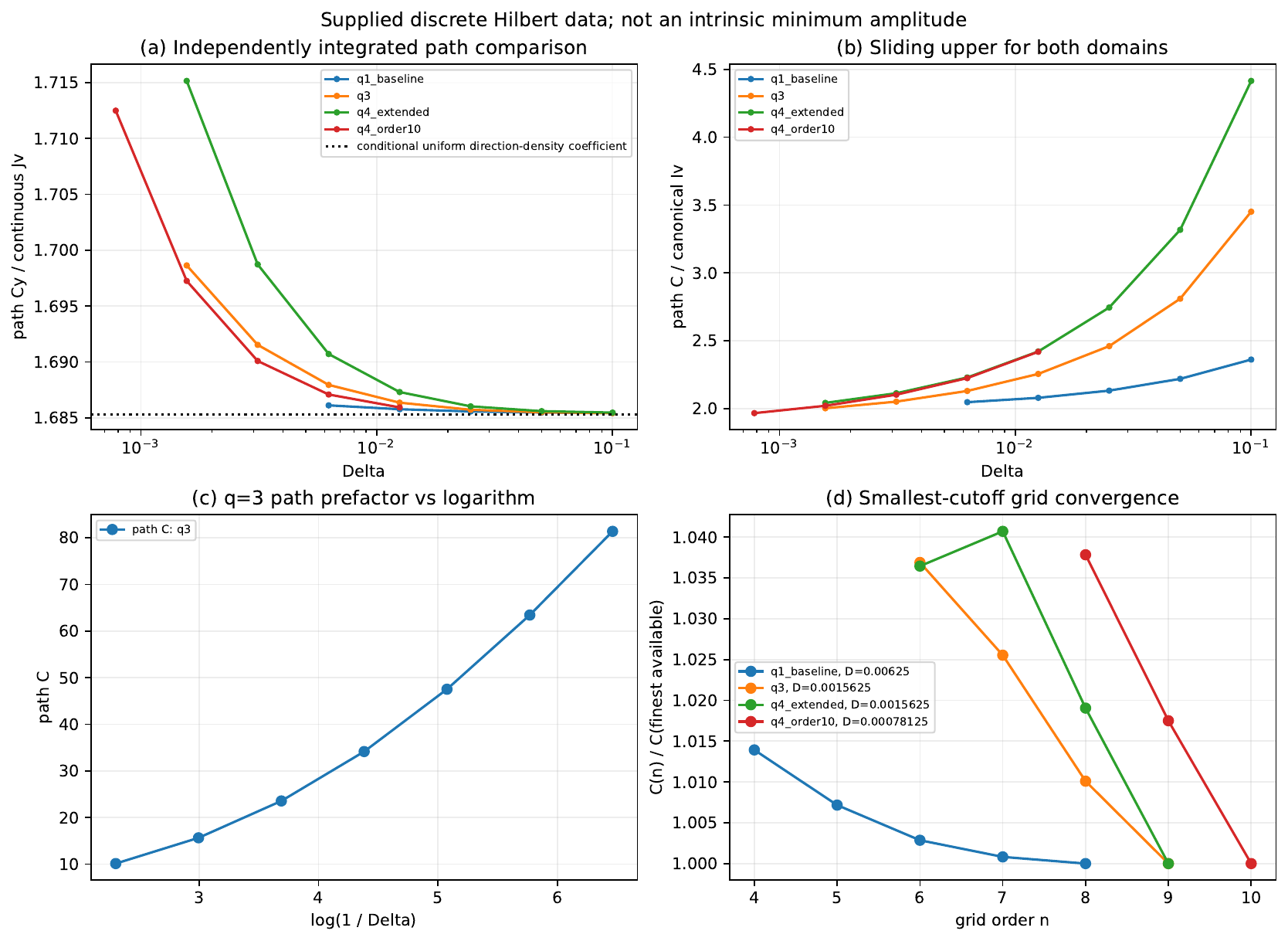}
    \caption{\HL_REF_AB{Independent-integral comparison for the prescribed coordinate-Hilbert family.  (a) Singular-direction path amplitude divided by the independently integrated actual-metric density.  The dotted coefficient additionally assumes uniform directional step density; it is not an optimal covering constant or a consequence of comparability alone.  (b) Total path amplitude divided by the Taylor-kernel integral, with the sliding upper limit used for both domains.  (c) The \(q=3\) path amplitude plotted directly against \(\log(1/\Delta)\); curvature over the available range is retained rather than fitted away.  (d) Prefactors at successive orders, normalized by the finest available order, at each dataset's smallest cutoff.  Passing the core-resolution flag does not guarantee grid convergence.}}
    \label{fig:path_integral_validation}
\end{figure*}

\HL_REF_A{The continuous integral can be extended to smaller \(\Delta\) without generating prohibitively large traversal grids.  Figure~\ref{fig:continuum_asymptotics} shows the \(q=3\) integral against \(\log(1/\Delta)\) and the compensated \(q=4\) quantity \(\sqrt{\Delta}\,J_v\).  The latter develops a plateau consistent with the Taylor-kernel asymptote; its smallest-interval sliding-window exponent is \(0.500060\), for \(\Delta=5.96\times10^{-9}\) to \(1.19\times10^{-8}\).  The \(q=3\) local exponent decreases toward zero, as expected for logarithmic growth.  These are continuous-density checks of the analytic classification, not new high-order Hilbert measurements.  Numerical tolerances, fixed/sliding-window comparisons, and the path convergence table are given in Appendix~\ref{app:critical_validation}.}

\begin{figure*}[t]
    \centering
    \includegraphics[width=0.95\textwidth]{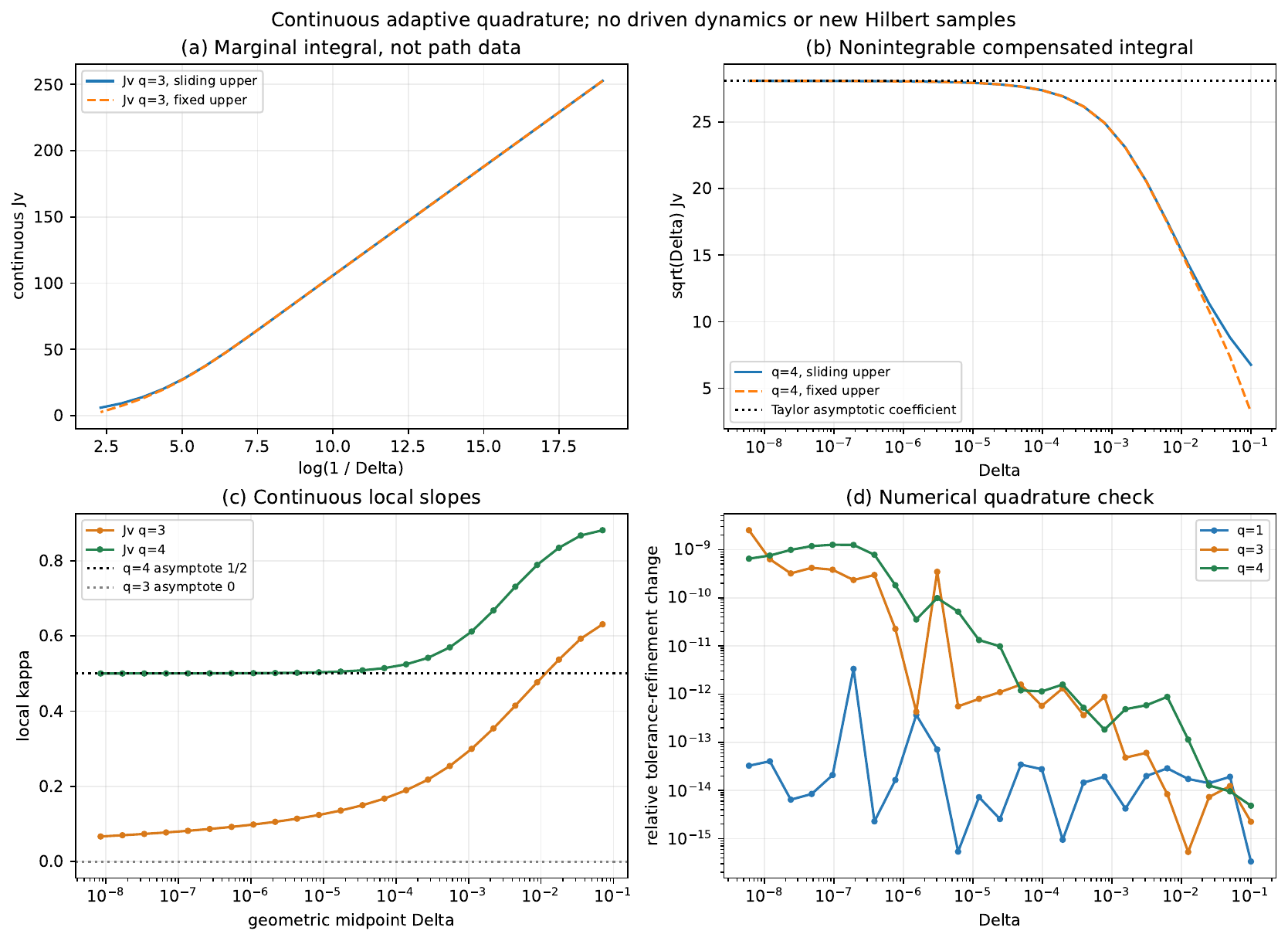}
    \caption{\HL_REF_A{Continuous adaptive-quadrature tests of the diagnostic integrals.  (a) The \(q=3\) actual-metric integral against \(\log(1/\Delta)\), with fixed and sliding upper limits distinguished.  (b) The compensated \(q=4\) integral approaches the leading Taylor coefficient.  (c) Local exponents of the actual-metric integrals, with logarithmic and square-root asymptotes indicated.  (d) Relative change under quadrature-tolerance refinement.  These curves contain no new discrete traversal data and do not represent a microscopic critical friction prediction.}}
    \label{fig:continuum_asymptotics}
\end{figure*}

\HL_REF_AB{For the Landau-response diagnostic geometries, the critical core has \(h\sim t^{3/2}\) and \(\chi^{-1}\sim t\).  The core contribution of the prescribed response-time weighted magnetic component \(g_{hh}=\chi^{1+z\nu}\) therefore has the scaling}
\begin{equation}
    \HL_REF_AB{I_h^{\rm core}(\Delta)
    \sim
    \int_\Delta^{t_0}\dd t\;t^{3/2}\,t^{-(1+z\nu)/2}.}
    \label{eq:ising_directional_integral_scaling}
\end{equation}
\HL_REF_AB{This core contribution is integrable for \(z\nu<4\), marginal at \(z\nu=4\), and divergent as \(\Delta^{-(z\nu/2-2)}\) for \(z\nu>4\).  The static diagnostic \(g_{hh}=\beta\chi\) and the response-time weighted diagnostic with \(z\nu=1\) both have finite directional integrals on removal of the cutoff.  Their regular thermal component and zero cross term remain modelling choices, so this is a classification of the specified geometries, not a universal conclusion about physical Ising dissipation.}

\subsection{\HL_REF_A{Ordering/connector benchmark with Morton/Z-order}}
\label{subsec:morton_control_results}

The lower-bound exponent is independent of a particular ordering, but finite-resolution amplitudes can depend on how local the ordering and its connectors are.  We therefore repeated the two-dimensional analysis with Morton/Z-order traversal as a control.  Morton order visits the same grid cells as the Hilbert traversal but introduces nonlocal jumps at recursive block boundaries.  \HL_REF_A{Every jump is evaluated with the midpoint straight coordinate-segment estimator in Eq.~\eqref{eq:analysis_discrete_length}.  The resulting comparison is an ordering/connector benchmark, not a unique physical cost assigned to the Morton ordering.  An alternative interpolation or time allocation would define a different protocol.}

Figure~\ref{fig:morton_locality_control} summarizes the comparison.  \HL_REF_B{After the Landau normalization correction, the modelwise mean Hilbert exponents range from \(0.9981\) to \(1.0059\), and the mean Morton exponents from \(1.0138\) to \(1.0352\).}  Thus Morton order produces a slightly larger finite-window exponent, but the result remains close to the two-dimensional value \(p=1\).  The fitted prefactor ratio \(C_{\rm M}/C_{\rm H}\), by contrast, is systematically larger than unity.  \HL_REF_B{Its modelwise mean ranges from \(1.121\) to \(1.749\).}  \HL_REF_A{Within the specified estimator and tested windows, the larger difference is in the amplitude.}

The cutoff-dependent ratios in Fig.~\ref{fig:morton_locality_control}(c) show the same effect in the standard critical windows.  The nonintegrable \(q=4\) stress test in Fig.~\ref{fig:morton_locality_control}(d) exhibits prefactor growth for both traversals.  At the highest order, the ratio \(C_{\Delta,\rm M}^{\rm hi}/C_{\Delta,\rm H}^{\rm hi}\) ranges from \(1.207\) to \(1.327\) over the simulated cutoff window.  The ratio decreases toward smaller \(\Delta\), indicating that as the singular \(v\)-direction length density dominates the total length, the relative contribution of Morton inter-block jumps becomes less important.  The local slopes at the smallest cutoff interval are \(\kappa_{\rm eff}\simeq0.681\) for Hilbert and \(\kappa_{\rm eff}\simeq0.617\) for Morton, both still finite-window estimates relative to the asymptotic \(q=4\) value \(1/2\).

\begin{figure*}[t]
    \centering
    \includegraphics[width=0.95\textwidth]{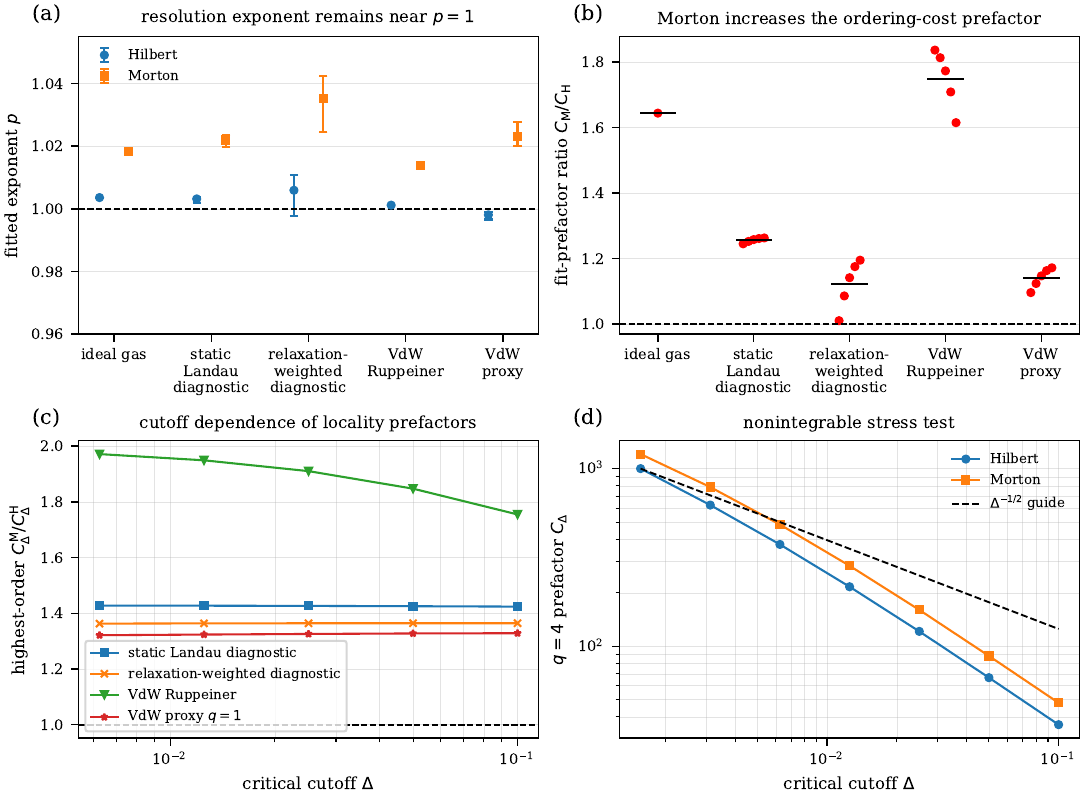}
    \caption{\HL_REF_A{Morton/Z-order as an ordering/connector benchmark for the specified midpoint coordinate-segment estimator.}  (a) Fitted resolution exponents for Hilbert and Morton traversals across the baseline model groups.  Both remain close to \(p=1\).  (b) Fitted prefactor ratio \(C_{\rm M}/C_{\rm H}\), showing systematic amplitude enhancement for Morton ordering.  (c) Highest-order prefactor ratios in the standard critical windows.  (d) Nonintegrable van der Waals response-proxy stress test with \(q=4\), where both traversals show critical prefactor growth while Morton retains a larger amplitude.  \HL_REF_B{The Landau entries use the revised diagnostic metrics.}}
    \label{fig:morton_locality_control}
\end{figure*}

The ideal-gas subdivision check gives \(p_{\rm M}=1.0183\) and \(C_{\rm M}/C_{\rm H}=1.644\) with one midpoint segment per jump, compared with \(p_{\rm M}=1.0160\) and \(C_{\rm M}/C_{\rm H}=1.662\) \HL_REF_A{when every jump in both traversals is split into four equal coordinate subsegments.  This small change indicates that the amplitude difference in this regular benchmark is not primarily a one-midpoint quadrature effect; it does not remove dependence on the chosen connector model.}

\subsection{Dimensional dependence of the resolution exponent}
\label{subsec:dimensional_dependence}

\HL_REF_AB{The three-dimensional calculation is a further benchmark of the known geometric exponent \(p=d-1\).}  We therefore repeated the analysis for a three-dimensional flat benchmark, the three-coordinate ideal-gas Ruppeiner metric \(g=\mathrm{diag}(c_1/x_1^2,c_2/x_2^2,c_3/x_3^2)\), which is the natural extension of Eq.~\eqref{eq:ideal_metric} and is brought to Euclidean form by \(y_i=\sqrt{c_i}\,\ln x_i\).  This provides a curvature-free \(d=3\) window with a closed-form Riemannian volume against which the prefactor can be checked.  The three-dimensional Hilbert traversal is generated by Skilling's arbitrary-dimension transform between Hilbert indices and integer coordinates \cite{Skilling2004}; as in two dimensions, consecutive sub-cubes share a face, so the traversal is locality preserving.  \HL_REF_AB{A boustrophedon raster supplies a second local construction.  Both traverse \(O(\delta^{-3})\) cells by steps of size \(O(\delta)\), giving the same \(O(\delta^{-2})\) upper scaling as the covering lower bound.}

The numerical conventions are identical to the two-dimensional study: a midpoint-rule discrete length, a coordinate covering radius \(\eps_n\) equal to the half-diagonal of a grid cell, and a log--log ordinary-least-squares slope fitted over the highest available orders.  Figure~\ref{fig:dimensional_dependence} shows the result.  After normalizing by the fitted prefactor, the three-dimensional lengths follow the reference \(\eps^{-2}\) law over the available orders, and the fitted exponents are \(p=2.007\pm0.003\) for the Hilbert traversal and \(p=2.011\pm0.002\) for the boustrophedon raster, with \(R^2>0.99999\) in both cases.  \HL_REF_B{The uncertainties are regression standard errors over orders \(n=4,5,6\), not estimates of the discretization error.}  \HL_REF_AB{Together with the two-dimensional benchmarks, the agreement illustrates the established dimensional exponent while leaving physical dissipation validation to the driven-trap calculation.}

\begin{figure*}[t]
    \centering
    \includegraphics[width=0.8\textwidth]{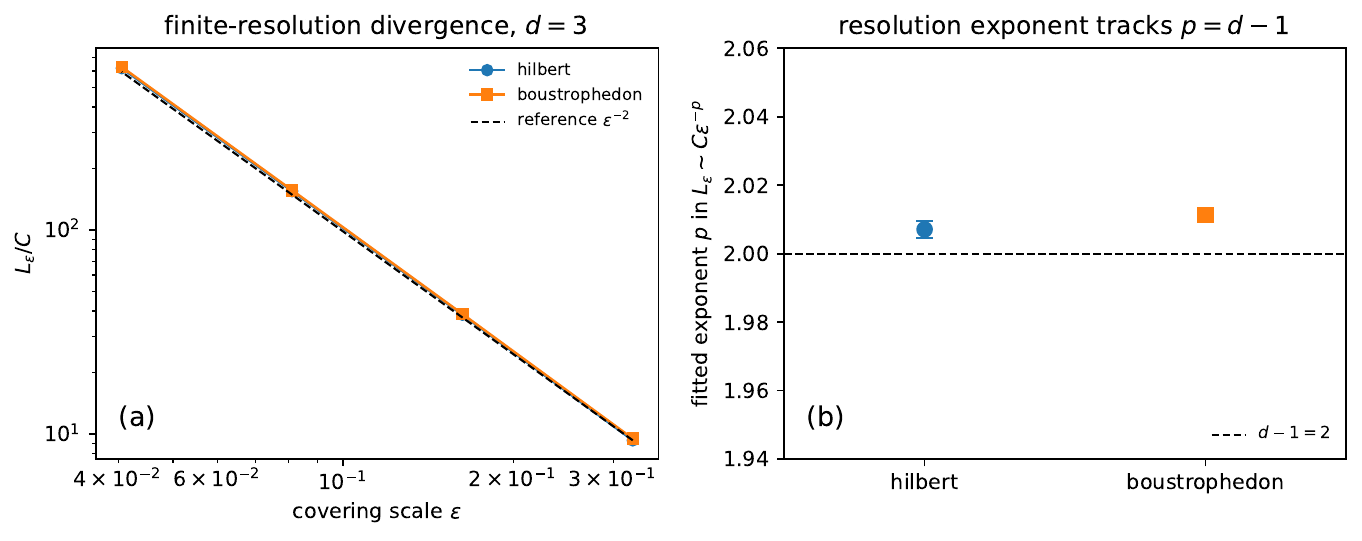}
    \caption{Dimensional dependence of the finite-resolution exponent, tested on the flat three-dimensional ideal-gas Ruppeiner benchmark.  (a) Normalized lengths \(L_\eps/C\) for the Skilling Hilbert traversal and a boustrophedon raster follow the reference \(\eps^{-2}\) law expected for \(d=3\).  (b) Fitted exponents \(p\) in \(L_\eps\sim C\eps^{-p}\) for both traversals, lying within a few parts in \(10^3\) of the predicted \(p=d-1=2\) (dashed line).  The exponent is fixed by the window dimension and is independent of the locality-preserving ordering.}
    \label{fig:dimensional_dependence}
\end{figure*}

\section{Discussion}
\label{sec:discussion}

\HL_REF_AB{The results connect three distinct levels of description. The geometric level is a classical coverage estimate: a curve that resolves a positive-volume regular $d$-dimensional window has length at least of order $\eps^{1-d}$, with local grid constructions giving the same exponent. The numerical values near $p=d-1$ are benchmarks of this estimate and of the discretization, rather than evidence for a new geometric exponent. At the thermodynamic level, combining this coverage requirement with an independently justified friction metric produces a coverage-specific duration--quadratic-work constraint. At the operational level, the prescribed trap raster shows how a serial scan allocates this cost between time and work.}

\HL_REF_B{These statements concern a single continuous protocol required to cover the whole prescribed window. They do not imply that every inference, optimization or equation-of-state estimation task must perform such a scan. A task that permits selective queries, parallel preparation or a lower-dimensional target may have a different coverage requirement. The Gaussian trap example makes the present requirement concrete: coverage in its friction metric controls proximity in the equilibrium mean and standard deviation. \HL_REF_A2{Equilibration and measurement times, statistical error, finite-time bias and temporal correlations require separate accounting, so the traversal-duration bound is not a complete acquisition-time estimate.}}

\HL_REF_AB{The distinction between geometric length and physical work is essential. For a specified path and timing, the quadratic action satisfies $W_{\rm ex}^{(2)}\ge L_\zeta^2/\tau$, with equality at constant friction speed. Fixed-$\tau$ refinement can leave the slow-driving regime; the resulting divergence of the action is therefore not a universal divergence theorem for exact excess work. The necessary quadratic-budget condition $\tau\ge L_\zeta^2/W_0$ becomes operational only together with control of local relaxation and of the time spent on connectors and turns. The independent moment calculation in Fig.~\ref{fig:trap_dynamic_validation} tests this approximation for explicitly specified noncritical dynamics. Its convergence to the quadratic prediction under slower driving supplies evidence beyond the algebraic doubling of the length exponent.}

\HL_REF_AB{For critical-window diagnostics, the order of limits and the protocol family remain part of the statement. In the present two-dimensional grid family we write}
\begin{equation}
    \HL_REF_AB{L_\eps(\Delta)\sim C_\Delta\eps^{-1},
    \qquad \eps\to0\ \text{at fixed }\Delta>0.}
    \label{eq:discussion_combined_scaling}
\end{equation}
\HL_REF_A{The subsequent cutoff dependence of $C_\Delta$ concerns the selected coordinate mesh, directional step density and connector prescription. It is not an intrinsic optimal covering coefficient. A passive coordinate change of the same geometric path preserves its length, whereas constructing a new uniform grid in new coordinates changes the protocol. The fixed-window exponent is stable under regular metric comparison; constants that depend on the cutoff cannot be discarded when discussing the limit $\Delta\to0$.}

\HL_REF_AB{The independent integral comparisons separate three effects: finite grid order in the path estimate, finite-cutoff corrections in the actual diagnostic metric, and the local Taylor approximation used for the asymptotic classification. For $q=3$, the logarithmic plot supplements the drifting local slopes. For $q=4$, compensated continuous integrals approach a plateau at smaller cutoffs than are reached by the path calculations. These extended integrals are not new high-order Hilbert measurements. The Landau-response metrics likewise remain prescribed diagnostic geometries; the normalization correction and positive-definite regular component do not constitute a microscopic derivation of critical Ising dissipation.}

\HL_REF_A{Morton results have the same status of a specified ordering/connector benchmark throughout: changing the physical interpolation of its nonlocal jumps would change the measured amplitudes. The calculations do not solve the minimum-constant coverage problem.}
\HL_REF_AB{Further work could compare adaptive meshes and apparatus-constrained connectors, or derive full friction tensors from interacting critical dynamics. Such extensions would require their own control variables, fluctuation normalization and limit ordering. They are separate from the regular-window coverage result established here.}

\section{Conclusion}
\label{sec:conclusion}

\HL_REF_AB{Exhaustive serial coverage of a compact regular thermodynamic control window imposes a length cost with exponent $d-1$. This is an application of classical geometric coverage estimates, not a new geometric exponent. \HL_REF_A2{Its thermodynamic consequence follows when the chosen coverage metric is a physical friction tensor or is uniformly dominated by one: within the slow-driving quadratic description, refining coverage imposes a growing lower bound on traversal duration at a fixed quadratic excess-work budget.} A fixed-duration divergence of the quadratic action must not be extrapolated to exact work outside that regime.}

\HL_REF_B{The overdamped trap provides a concrete check using independently integrated Gaussian moment dynamics. With the path, connectors and timing fixed, the dynamic excess work approaches the quadratic estimate as the scan is slowed. A family whose duration grows as $L_\zeta^2$ maintains a controlled work level over the tested refinements, subject to the reported relaxation diagnostics.}

\HL_REF_AB{A quasi-uniform operational resolution floor regularizes the geometric cost. \HL_REF_A2{The critical response-proxy calculations characterize the integrability of directional length densities in the prescribed diagnostic geometries; their coordinate- and protocol-dependent prefactors are neither intrinsic optimal covering constants nor universal predictions of microscopic critical dissipation.} Independent quadrature and grid-order tables strengthen this conditional interpretation. Together, these results specify when whole-window coverage can be translated into a thermodynamic resource requirement and where additional microscopic or measurement modelling is still needed.}

\section*{Acknowledgments}
This work was supported by JSPS KAKENHI (Grant No. 22K14177) and JST PRESTO (Grant No. JPMJPR23O7). The author gratefully acknowledges the laboratory staff for their valuable technical and administrative support throughout this work. The author is also deeply grateful to his family for their unwavering support and encouragement. ChatGPT (GPT-5.6, OpenAI) was used to provide initial English-language drafting assistance during manuscript preparation. All AI-assisted text was critically reviewed and revised by the author, who takes full responsibility for the accuracy, integrity, and originality of the final manuscript.


\clearpage
\appendix

\section{Numerical parameter values, fit diagnostics, and pseudocode}
\label{app:numerical_details}

\HL_REF_AB{The appendices collect model parameters, numerical procedures and fit diagnostics. Sections~\ref{sec:state_space}--\ref{sec:conclusion} separate geometric coverage, microscopic noncritical dynamics and prescribed critical diagnostic geometries. The harmonic-trap friction convention follows Ref.~\cite{SivakCrooks2012}.}

\HL_REF_A{This appendix gives the parameter values, fit diagnostics, and pseudocode supporting Sec.~\ref{sec:analysis}.}  \HL_REF_B{The geometric benchmark traversals and their fits are deterministic.  They use no random sampling or ensemble averaging; the independent driven-trap calculation uses closed moment equations rather than sampled trajectories.}

For an order-$n$ rectangular grid, the number of cells is $2^n\times 2^n$.  If the coordinate window is $[x_-,x_+]\times[y_-,y_+]$, the cell widths are
\begin{equation}
    \Delta x_n=\frac{x_+-x_-}{2^n},\qquad
    \Delta y_n=\frac{y_+-y_-}{2^n},
\end{equation}
and the coordinate covering radius reported in the figures is
\begin{equation}
    \eps_n=\frac{1}{2}\sqrt{\Delta x_n^2+\Delta y_n^2}.
    \label{eq:appendix_coordinate_epsilon}
\end{equation}
For each fixed regular cutoff window this coordinate scale is uniformly comparable to the intrinsic Riemannian covering radius, as discussed around Eq.~\eqref{eq:numerical_epsilon_comparison}.

\begin{table*}[t]
\caption{Model parameters, coordinate windows, and metric components used in the numerical analysis.  Here $t=T-T_c$, $\beta=1/T$, $m$ is the real solution of $b m^3+a t m-h=0$, $\chi^{-1}=a t+3bm^2$, and $s(T,v)=(v-b)^{-2}-2a/(T v^3)$.  \HL_REF_B{The two Landau-response metrics are prescribed diagnostic geometries, with $\tau_T=\tau_0=1$ in the weighted case.}}
\label{tab:appendix_model_windows}
\begin{ruledtabular}
\setlength{\tabcolsep}{3pt}
\begin{tabular}{p{0.17\textwidth}p{0.13\textwidth}p{0.25\textwidth}p{0.18\textwidth}p{0.21\textwidth}}
Case & Coordinates & Coordinate window & Parameters & Metric components \\
\hline
Ideal gas & $(u,v)$ & $u\in[1,4]$, $v\in[1,4]$ & $c_v=1.5$ & $g_{uu}=c_v/u^2$, $g_{vv}=1/v^2$ \\
\HL_REF_B{Landau static diagnostic} & $(t,h)$ & $t\in[\Delta,\Delta+0.25]$, $h\in[-0.2,0.2]$ & $T_c=1$, $a=b=C=1$ & \HL_REF_B{$g^{\rm st}_{tt}=C/T^2$, $g^{\rm st}_{hh}=\beta\chi$} \\
\HL_REF_B{Landau relaxation-weighted diagnostic} & $(t,h)$ & $t\in[\Delta,\Delta+0.25]$, $h\in[-0.2,0.2]$ & $T_c=1$, $a=b=C=1$, $z\nu=1$ & \HL_REF_B{$g^{\rm wt}_{tt}=C/T^2$, $g^{\rm wt}_{hh}=\chi^{1+z\nu}$} \\
VdW Ruppeiner pullback & $(T,v)$ & $T\in[T_c+\Delta,T_c+\Delta+0.15]$, $v\in[v_c-0.5,v_c+0.5]$ & $a=b=1$, $c_v=1.5$, $T_c=8/27$, $v_c=3$ & $g_{TT}=c_v/T^2$, $g_{vv}=s(T,v)$ \\
VdW response proxy & $(T,v)$ & $T\in[T_c+\Delta,T_c+\Delta+0.15]$, $v\in[2.5,3.5]$ & $a=b=1$, $c_v=1.5$, $q=1,3,4$ & $g_{TT}=c_v/T^2$, $g_{vv}=s(T,v)^{-q}$ \\
Harmonic-trap operational raster & $(a,k)$ & $a\in[0,1]$, $k\in[1,4]$ & $\beta=\mu=1$ & $\zeta_{aa}=1$, $\zeta_{kk}=(4k^3)^{-1}$ \\
\end{tabular}
\end{ruledtabular}
\end{table*}

\begin{table*}[t]
\caption{Grid orders, fit windows, cutoff values, and prefactor estimators.  The notation $4$--$8$ means all integer orders in that range.  The default critical cutoffs are $\Delta=0.1,0.05,0.025,0.0125,0.00625$.}
\label{tab:appendix_run_conditions}
\begin{ruledtabular}
\setlength{\tabcolsep}{3pt}
\begin{tabular}{p{0.22\textwidth}p{0.16\textwidth}p{0.13\textwidth}p{0.13\textwidth}p{0.28\textwidth}}
Data set & Traversal and model class & Computed orders & Fit orders & Cutoffs and estimator \\
\hline
Baseline scaling, Table~\ref{tab:baseline_exponents}, and Fig.~\ref{fig:standard_cutoff_behavior} & Hilbert; all standard models & $4$--$8$ & $5$--$8$ & Default critical cutoffs where applicable; $C_\Delta^{\rm hi}=L_{\eps_{\min}}\eps_{\min}$ from order $8$ \\
$q=3$ marginal VdW proxy & Hilbert; VdW response proxy & $6$--$9$ & $6$--$9$ & $\Delta=0.1$ to $0.0015625$ by factors of two; highest-order prefactor from order $9$ \\
$q=4$ nonintegrable VdW proxy & Hilbert; VdW response proxy & $6$--$9$ & $6$--$9$ & $\Delta=0.1$ to $0.0015625$ by factors of two; highest-order prefactor from order $9$ \\
$q=4$ small-$\Delta$ check & Hilbert; VdW response proxy & $8$--$10$ & $8$--$10$ & $\Delta=0.0125$ to $0.00078125$ by factors of two; highest-order prefactor from order $10$ \\
Baseline locality comparison & Hilbert and Morton/Z-order; all standard models & $4$--$8$ & $5$--$8$ & Default critical cutoffs where applicable; straight coordinate-segment transition-cost estimator \\
$q=4$ locality stress test & Hilbert and Morton/Z-order; VdW response proxy & $6$--$9$ & $6$--$9$ & $\Delta=0.1$ to $0.0015625$ by factors of two; straight coordinate-segment transition-cost estimator \\
Ideal-gas subdivision check & Hilbert and Morton/Z-order; ideal gas & $4$--$8$ & $5$--$8$ & \HL_REF_A{Every jump in both traversals subdivided into four equal coordinate subsegments before midpoint quadrature} \\
\HL_REF_B{Three-dimensional benchmark} & \HL_REF_B{Hilbert and boustrophedon; flat ideal-gas extension} & \HL_REF_B{$3$--$6$} & \HL_REF_B{$4$--$6$} & \HL_REF_B{$x_i\in[1,4]$, $(c_1,c_2,c_3)=(1.5,1,1)$; coordinate cell half-diagonal and midpoint length} \\
Operational harmonic-trap raster & Serpentine raster; harmonic-trap friction metric & $3$--$9$ & $5$--$9$ & $a\in[0,1]$, $k\in[1,4]$, $\beta=\mu=1$, $\tau=1$, $\HL_REF_B{t_{\rm move}}=1$; \HL_REF_B{quadratic estimates at fixed total time or equal moving-segment durations} \\
\end{tabular}
\end{ruledtabular}
\end{table*}

The exponent fits use
\begin{equation}
    \log L_n=\alpha-p\log\eps_n+r_n,
    \qquad \alpha=\log C,
    \label{eq:appendix_log_fit}
\end{equation}
with ordinary least squares over the fit orders listed in Table~\ref{tab:appendix_run_conditions}.  The quoted uncertainty in $p$ is the standard error of the fitted log--log slope.  It is therefore a regression diagnostic over the chosen grid orders, not a bootstrap error bar or an uncertainty from independent stochastic repetitions.  The reported residual diagnostics use the log residuals $r_n$ and the relative residuals $\exp(r_n)-1$.  \HL_REF_B{Neither the slope standard error nor the residual is a discretization uncertainty.  The order-by-order amplitudes in Table~\ref{tab:critical_grid_convergence} provide a separate grid-refinement comparison at fixed cutoff.}

For the operational harmonic-trap raster, the same log--log fitting convention is applied to four observables: the friction length $L_\zeta$, the fixed-total-time quadratic estimate $W_{\rm fixed\;time}^{(2)}=L_\zeta^2/\tau$, the \HL_REF_B{estimate for equal moving-segment durations} $\HL_REF_B{W_{\rm move}^{(2)}}=\sum_i\ell_i^2/\HL_REF_B{t_{\rm move}}$, and the \HL_REF_B{total moving time} $\HL_REF_B{\tau_{\rm move}}=N_{\rm seg}\HL_REF_B{t_{\rm move}}$.  The fitted exponents reported in Table~\ref{tab:operational_harmonic_fits} are regression diagnostics over the tail orders $n=5,\ldots,9$.  \HL_REF_B{The former ``dwell'' convention assigned time to motion along a segment.  The revised symbol $t_{\rm move}$ makes explicit that no additional holding, equilibration, or measurement interval is included.  Since the fixed-time work estimate is defined as $L_\zeta^2/\tau$, its fitted exponent is algebraically twice the length exponent; the separate driven-moment validation is described in Appendix~\ref{sm:trap_dynamic_details}.}

\begin{table*}[t]
\caption{Baseline Hilbert fit diagnostics by model family.  The table reports the largest standard error of $p$, the smallest log-space $R^2$, and the largest absolute relative residual over the fitted tail orders and, where applicable, over the listed critical cutoffs.}
\label{tab:appendix_baseline_residuals}
\begin{ruledtabular}
\setlength{\tabcolsep}{3pt}
\begin{tabular}{p{0.30\textwidth}cccc}
Model family & Range of fitted $p$ & Max. s.e. of $p$ & Min. $R^2$ & Max. relative residual \\
\hline
Ideal gas & 1.003609 & 0.001135 & 0.999997 & 0.158\% \\
\HL_REF_B{Landau static diagnostic} & 1.001886--1.003869 & 0.000862 & 0.999999 & 0.111\% \\
\HL_REF_B{Landau relaxation-weighted diagnostic, $z\nu=1$} & \HL_REF_B{0.997839--1.010883} & \HL_REF_B{0.003977} & \HL_REF_B{0.999968} & \HL_REF_B{0.606\%} \\
VdW Ruppeiner pullback & 1.001053--1.001324 & 0.000439 & $>0.999999$ & 0.062\% \\
VdW response proxy, $q=1$ & 0.996616--0.999141 & 0.000798 & 0.999999 & 0.102\%
\end{tabular}
\end{ruledtabular}
\end{table*}

\begin{table*}[t]
\caption{Additional fit and prefactor diagnostics for the marginal, nonintegrable, and Morton/Z-order runs.  ``Full tail'' means the fit orders in Table~\ref{tab:appendix_run_conditions}; ``core-resolved'' applies the critical-core resolution flag described below.}
\label{tab:appendix_additional_residuals}
\begin{ruledtabular}
\setlength{\tabcolsep}{3pt}
\begin{tabular}{p{0.28\textwidth}p{0.13\textwidth}p{0.15\textwidth}p{0.16\textwidth}p{0.17\textwidth}}
Data set & Mean or range of $p$ & Prefactor diagnostic & Max. relative residual & Interpretation \\
\hline
VdW proxy $q=3$ & 0.982131--0.998917 & dynamic range $8.03$; smallest-interval local slope $0.359$ & 0.316\% & Marginal by the analytic integrability criterion; finite-window slopes drift downward \\
VdW proxy $q=4$, extended & 0.973703--0.998473 & dynamic range $27.68$; smallest-interval local slope $0.681$ & 0.989\% full tail; 0.441\% core-resolved & Nonintegrable stress test; finite-window drift toward $1/2$ \\
VdW proxy $q=4$, order-10 check & 0.973209--0.997682 & dynamic range $7.06$; smallest-interval local slope $0.624$ & 0.279\% & Smaller-cutoff check of the same drift \\
Baseline Morton comparison & \HL_REF_B{Morton $p$ range $1.013127$--$1.042327$ across standard families} & \HL_REF_B{mean $C_{\rm M}/C_{\rm H}=1.121$--$1.749$ depending on family} & \HL_REF_B{0.982\%} & \HL_REF_A{Near-$\eps^{-1}$ benchmark exponent; connector-dependent amplitude} \\
$q=4$ Morton stress test & Hilbert/Morton; order $6$--$9$ & $C_{\rm M}^{\rm hi}/C_{\rm H}^{\rm hi}=1.207$--$1.327$ & 3.73\% full tail; 1.82\% core-resolved & Larger finite-window residuals because the critical prefactor is rapidly varying \\
Ideal-gas Morton subdivision & Morton with four subsegments per jump & $p_{\rm M}=1.015987$, $C_{\rm M}/C_{\rm H}=1.662$ & 0.455\% & Morton/Hilbert amplitude difference is not primarily a midpoint-quadrature artifact
\end{tabular}
\end{ruledtabular}
\end{table*}

Critical-core resolution was monitored by a dimensionless score
\begin{equation}
    \rho_\Delta=\max\left(\frac{\Delta x_n}{x_{\rm core}(\Delta)},
    \frac{\Delta y_n}{y_{\rm core}(\Delta)}\right),
    \label{eq:appendix_core_score}
\end{equation}
with threshold $\rho_\Delta>0.5$ for an under-resolved critical core.  For the Landau Ising windows we used $x_{\rm core}\sim\Delta$ in the temperature-like direction and $y_{\rm core}\sim\Delta^{3/2}$ in the field direction.  For the van der Waals windows we used $x_{\rm core}\sim\Delta$ in the $T-T_c$ direction and $y_{\rm core}\sim\Delta^{1/2}$ in the $v-v_c$ direction.  With this criterion, the two smallest baseline Ising cutoffs, $\Delta=0.0125$ and $0.00625$, are flagged as under-resolved at the highest baseline order.  They are retained in the plots as finite-resolution diagnostics but are not used to assign an asymptotic critical exponent.  \HL_REF_B{Conversely, a passed core flag only tests a geometric scale criterion and does not certify convergence of the amplitude.}

\noindent\textbf{Pseudocode for the numerical length and fit estimates.}
\begin{enumerate}
    \item Select a model, a cutoff $\Delta$ if required, and a rectangular coordinate window from Table~\ref{tab:appendix_model_windows}.  Define the metric $g(x)$ on that window.
    \item Select a traversal type $\mathcal T$.  For Hilbert traversal, order the $2^n\times2^n$ cell centers by the finite Hilbert ordering.  For Morton traversal, order the same centers by the bit-interleaved Morton key in Eq.~\eqref{eq:morton_key}.
    \item For each order $n$ listed in Table~\ref{tab:appendix_run_conditions}, compute the cell centers $x_0,\ldots,x_N$ and the coordinate covering radius $\eps_n$ from Eq.~\eqref{eq:appendix_coordinate_epsilon}.
    \item Initialize $L_n=0$.  For each consecutive pair $(x_k,x_{k+1})$, set $\Delta x_k=x_{k+1}-x_k$ and $x_{k+1/2}=(x_k+x_{k+1})/2$, then add
    \[
        \left[\Delta x_k^i g_{ij}(x_{k+1/2})\Delta x_k^j\right]^{1/2}
    \]
    to $L_n$.  \HL_REF_A{For the ideal-gas subdivision check, split each jump in both traversals into four equal coordinate subsegments and apply the midpoint rule to each subsegment.}
    \item Fit $\log L_n=\alpha-p\log\eps_n+r_n$ over the specified fit orders.  Store $p$, its slope standard error, the log-space $R^2$, and the residuals $r_n$.
    \item For each critical cutoff, compute the highest-order prefactor $C_\Delta^{\rm hi}=L_{\eps_{\min}}\eps_{\min}$.  Estimate local critical slopes from adjacent values of $\log C_\Delta^{\rm hi}$ versus $\log\Delta$.
    \item Apply the core-resolution score in Eq.~\eqref{eq:appendix_core_score} to identify finite-resolution cutoffs whose critical core is not resolved by the grid.
\end{enumerate}

\noindent\textbf{Additional pseudocode for the operational harmonic-trap scan.}
\begin{enumerate}
    \item Set $a\in[0,1]$, $k\in[1,4]$, $\beta=\mu=1$, and use the friction metric $\zeta_{aa}=1$, $\zeta_{kk}=(4k^3)^{-1}$.
    \item For each order $n=3,\ldots,9$, form a $2^n\times2^n$ serpentine raster over cell centers.  The covering scale is the coordinate half-diagonal.
    \item For each segment, compute its midpoint friction length $\ell_i$ and accumulate $L_\zeta=\sum_i\ell_i$ and $\sum_i\ell_i^2$.
    \item Report the quadratic fixed-time estimate $W_{\rm fixed\;time}^{(2)}=L_\zeta^2/\tau$ with $\tau=1$, and $\HL_REF_B{W_{\rm move}^{(2)}}=\sum_i\ell_i^2/\HL_REF_B{t_{\rm move}}$, $\HL_REF_B{\tau_{\rm move}}=N_{\rm seg}\HL_REF_B{t_{\rm move}}$ with $\HL_REF_B{t_{\rm move}}=1$.
    \item Fit $L_\zeta$, $W_{\rm fixed\;time}^{(2)}$, $\HL_REF_B{W_{\rm move}^{(2)}}$, and $\HL_REF_B{\tau_{\rm move}}$ as powers of the covering scale.
\end{enumerate}

\section{Standard critical-cutoff diagnostics}
\label{app:standard_diagnostics}

Figure~\ref{fig:standard_cutoff_behavior} collects the baseline cutoff diagnostics that support the interpretation of the critical-window results.  \HL_REF_AB{For the prescribed coordinate windows and traversal family, the static and relaxation-weighted Landau diagnostics with \(z\nu=1\), van der Waals Ruppeiner pullback, and van der Waals response proxy with \(q=1\) show only modest finite-window prefactor variation.}  The resolution exponent remains close to \(p=1\) throughout.  These diagnostics are used as controls for the nonintegrable \(q=4\) stress test in the main text.

\begin{figure*}[t]
    \centering
    \includegraphics[width=0.95\textwidth]{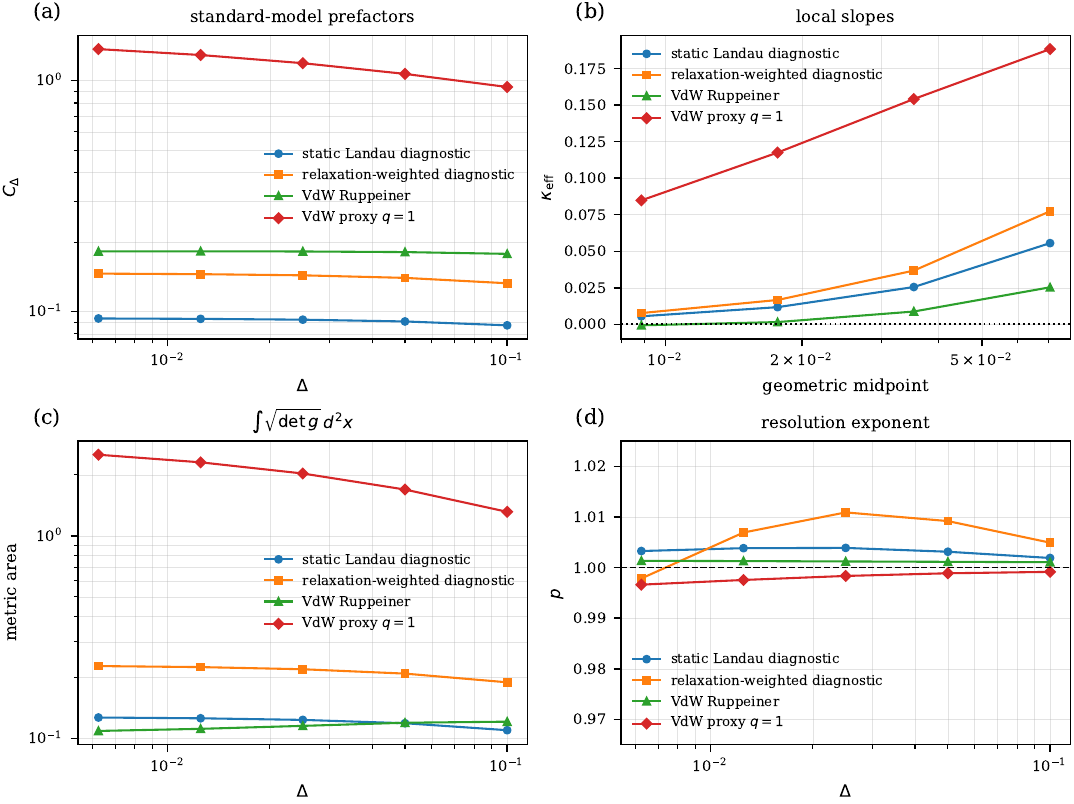}
    \caption{\HL_REF_AB{Baseline cutoff diagnostics for the prescribed geometries and coordinate protocol.}  (a) Highest-order prefactors \(C_\Delta\) for standard critical-window models.  (b) Local effective slopes for \(C_\Delta\).  (c) Riemannian metric-area estimates.  (d) Resolution exponents \(p\), which remain close to unity across the baseline cutoff windows.}
    \label{fig:standard_cutoff_behavior}
\end{figure*}

\HL_REF_B{The static/weighted Landau comparison in Fig.~\ref{fig:friction_comparability} samples each fixed window on a \(129\times129\) grid including its boundary. The reported component-ratio extrema are numerical estimates, not certified continuous bounds. Positive definiteness and compactness establish existence of comparison constants separately.}

\section{\HL_REF_B{Driven harmonic-trap protocols and numerical validation}}
\label{sm:trap_dynamic_details}

\HL_REF_B{The dynamic validation uses the exact Gaussian moment equations of Eqs.~\eqref{eq:trap_moment_dynamics}--\eqref{eq:trap_dynamic_work}, with the control window \(a\in[0,1]\), \(k\in[1,4]\), \(\beta=\mu=1\), and an equilibrium initial distribution. At order \(n\), the raster visits the \(2^n\times2^n\) cell centers, in alternating horizontal rows, with all vertical row-end connectors included. Endpoints are cell centers, so \(\Delta F\) is evaluated using the actual initial and final stiffness of each resolution. There is no final holding interval; \(W-\Delta F\) is the mean excess work even when the final state is not equilibrated.}

\HL_REF_B{Within horizontal segments, \(a(t)\) changes linearly and \(k\) is fixed. Within vertical segments, \(s(t)=k(t)^{-1/2}\) changes linearly and \(a\) is fixed. The controls and the mean and variance are continuous at every junction. The final mean and variance of one segment become the initial state of the next; the distribution is not reset to equilibrium at each cell. With the exact segment lengths of Eq.~\eqref{eq:trap_exact_segment_lengths}, the assigned times \(\Delta t_i=\tau\ell_i/L_{\zeta,\rm exact}\) give constant friction speed. Velocity directions may change at row turns.}

\HL_REF_B{Horizontal propagation, including its work integral, is analytic. For vertical segments the mean is also analytic, and the variance lag \(z=v-(\beta k)^{-1}\) is integrated with the adaptive Radau method. Both total work and excess work are accumulated separately. Relative and absolute integration tolerances are \(10^{-8}\) and \(10^{-10}\). The largest sampled discrepancy in \(W-\Delta F-W_{\rm ex}^{\rm dyn}\) across the 24 protocols is \(1.55\times10^{-10}\); this is an accounting diagnostic, not a rigorous error bound. Analytic fixed-stiffness cases and an independent integration of the untransformed mean, variance, and work equations check the implementation. No random-trajectory sampling is required for this Gaussian model.}

\HL_REF_AB{The 12 fixed-duration runs use \(n=3,4,5\) and \(\tau=1,10,100,1000\). The 12 prescribed-budget runs use the same resolutions and Eq.~\eqref{eq:trap_budget_family} with \(W_0=0.05\) and \(\alpha=1,4,16,64\). The latter has quadratic action \(W_0/\alpha\), rather than exactly \(W_0\) for every \(\alpha\). Table~\ref{tab:trap_dynamic_budget_diagnostics} reports representative results.}

\begin{table*}[t]
\caption{\HL_REF_AB{Prescribed-budget harmonic-trap diagnostics. \(R_\alpha=W_{\rm ex}^{\rm dyn}/(L_{\zeta,\rm exact}^2/\tau)\). The last column is the minimum segment duration relative to the slower, mean-relaxation time, evaluated conservatively with the minimum stiffness on that segment, for \(\alpha=1\). The durations grow as specified in Eq.~\eqref{eq:trap_budget_family}; the table tests the finite set of protocols, not a uniform theorem for all refinements.}}
\label{tab:trap_dynamic_budget_diagnostics}
\begin{ruledtabular}
\begin{tabular}{ccccc}
\HL_REF_B{\(n\)} & \HL_REF_B{\(L_{\zeta,\rm exact}\)} & \HL_REF_B{\(R_1\)} & \HL_REF_B{\(R_{64}\)} & \HL_REF_B{\(\min_i\mu k_{i,\min}\Delta t_i\), \(\alpha=1\)} \\
\hline
\HL_REF_B{3} & \HL_REF_B{7.405515} & \HL_REF_B{0.995476} & \HL_REF_B{0.999929} & \HL_REF_B{13.8546} \\
\HL_REF_B{4} & \HL_REF_B{15.450218} & \HL_REF_B{0.997808} & \HL_REF_B{0.999966} & \HL_REF_B{14.4771} \\
\HL_REF_B{5} & \HL_REF_B{31.474400} & \HL_REF_B{0.998913} & \HL_REF_B{0.999983} & \HL_REF_B{14.7518}
\end{tabular}
\end{ruledtabular}
\end{table*}

\HL_REF_B{The local diagnostics also sample \(|m-a|\sqrt{\beta k}\), \(|\beta k v-1|\), and variance positivity at nine equally spaced normalized-time nodes per segment, including endpoints. They are sampled extrema, not certified continuous-time maxima or minima. Diagnostic nodes are evaluated before stored trajectories are thinned to at most 400 nodes per protocol. For the \(\alpha=1\) budget scans, the largest sampled dimensionless mean lag and relative variance lag across all three resolutions are 0.00620 and 0.00612, respectively; both decrease when \(\alpha\) is increased. These checks supplement, rather than follow from, the duration law.}

\HL_REF_AB{A simple exact solution illustrates why a quadratic-action bound cannot be extrapolated to arbitrary driving speed. For an equilibrium initial state at fixed \(k\), with the center moved linearly by \(\Delta a\) in time \(T\),}
\begin{equation}
 \HL_REF_AB{W_{\rm ex}^{\rm dyn}
 =\frac{(\Delta a)^2}{\mu T}
 \left[1-\frac{1-e^{-\mu kT}}{\mu kT}\right]}.
 \label{eq:trap_fixed_k_exact_work}
\end{equation}
\HL_REF_AB{It approaches the quadratic value \((\Delta a)^2/(\mu T)\) when \(\mu kT\) is large, but approaches \(k(\Delta a)^2/2\) as \(T\to0\). Thus even this elementary finite-time process need not obey a lower bound given by the quadratic prediction. The raster validation uses agreement in the slow regime as its criterion.}

\section{\HL_REF_AB{Independent critical-integral and grid-convergence checks}}
\label{app:critical_validation}

\HL_REF_AB{This section separates the finite coordinate-grid paths from continuous-density calculations.  The path windows are \(\tau\in[\Delta,\Delta+0.15]\), \(u\in[-0.5,0.5]\), with \(a=b=1\), \(c_v=1.5\).  The actual density uses \(s^{-q/2}\); the canonical density uses \((A\tau+Bu^2)^{-q/2}\), where \(A=27/32\), \(B=1/48\).  The fixed-upper comparison instead uses \(\tau\in[\Delta,0.15]\).  Ratios across these two different domains are not interpreted as same-window convergence tests.}

\HL_REF_A{The actual stability factor is evaluated in a cancellation-resistant form,}
\begin{equation}
 \HL_REF_A{s(T_c+\tau,3+u)=
 \frac{\tau(3+u)^3+(2/27)u^2(9+4u)}
 {(T_c+\tau)(3+u)^3(2+u)^2}.}
 \label{eq:stable_stability_evaluation}
\end{equation}
\HL_REF_A{For each outer temperature, the variable \(u=\sqrt{\tau}\tan\theta\) resolves the narrow core, and the outer quadrature uses \(\log\tau\).  The full continuum check evaluates \(q=1,3,4\), both upper-limit conventions, and \(\Delta=0.1/2^j\) for \(j=0,\ldots,24\).  Relative and absolute quadrature tolerances are \(2\times10^{-9}\) and \(10^{-9}\); comparison with relative tolerance \(10^{-6}\) gives a maximum relative change \(2.53\times10^{-9}\).  The canonical integrals are also compared with their closed-form evaluations, agreeing to \(7.1\times10^{-15}\) relatively.  No quadrature warnings were returned.  These numerical indicators are not rigorous error bounds or statistical confidence intervals.}

\begin{figure*}[t]
 \centering
 \includegraphics[width=0.95\textwidth]{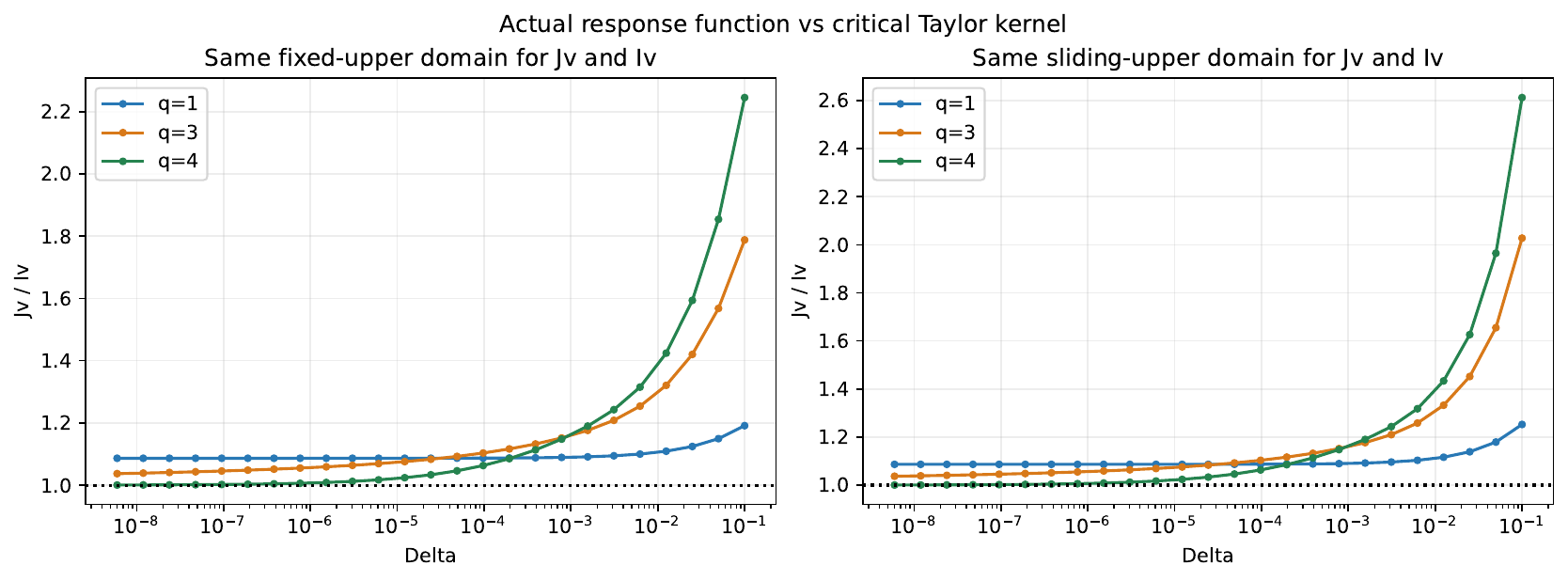}
 \caption{\HL_REF_A{Actual-metric and canonical directional-density integrals, evaluated independently of the traversal grid.  Fixed and sliding upper limits are displayed separately.  Their agreement in singular class does not eliminate finite-window differences in values and effective slopes.}}
 \label{fig:exact_vs_canonical_integrals}
\end{figure*}

\HL_REF_B{Table~\ref{tab:critical_grid_convergence} gives the highest-order check without mixing finite-cutoff drift and grid refinement.  At \(\Delta=0.00078125\), the order-8 core flag fails, while orders 9 and 10 pass.  Nevertheless, the order-9 to order-10 prefactor changes by \(1.72\%\) relative to order 9.  A core-resolution flag therefore identifies a minimum spatial criterion, not an uncertainty estimate or proof of convergence.  The regression standard errors reported elsewhere measure log--log fit scatter over the selected orders and are distinct from these finite-grid changes.}

\begin{table*}[t]
 \caption{\HL_REF_B{Grid refinement of the total \(q=4\) coordinate-Hilbert prefactor \(C_\Delta^{(n)}=\eps_n L_n\).  The last column is \(100[C_\Delta^{(10)}/C_\Delta^{(9)}-1]\), not an estimated continuum error.}}
 \label{tab:critical_grid_convergence}
 \begin{ruledtabular}
 \begin{tabular}{ccccc}
 \HL_REF_B{\(\Delta\)} & \HL_REF_B{order 8} & \HL_REF_B{order 9} & \HL_REF_B{order 10} & \HL_REF_B{relative 9--10 change} \\
 \hline
 \HL_REF_B{0.0125} & \HL_REF_B{216.936485} & \HL_REF_B{216.416902} & \HL_REF_B{216.240630} & \HL_REF_B{\(-0.0815\%\)} \\
 \HL_REF_B{0.003125} & \HL_REF_B{631.704502} & \HL_REF_B{624.386019} & \HL_REF_B{621.210479} & \HL_REF_B{\(-0.5086\%\)} \\
 \HL_REF_B{0.00078125} & \HL_REF_B{1585.273334} & \HL_REF_B{1554.208404} & \HL_REF_B{1527.475249} & \HL_REF_B{\(-1.7200\%\)}
 \end{tabular}
 \end{ruledtabular}
\end{table*}

\HL_REF_AB{The finite-path summary below is retained as a diagnostic of the original cutoff range.  The new independent integrals provide additional evidence for the analytic classification but do not turn the finite-path exponents into asymptotic measurements.}

\begin{figure*}[t]
    \centering
    \includegraphics[width=0.95\textwidth]{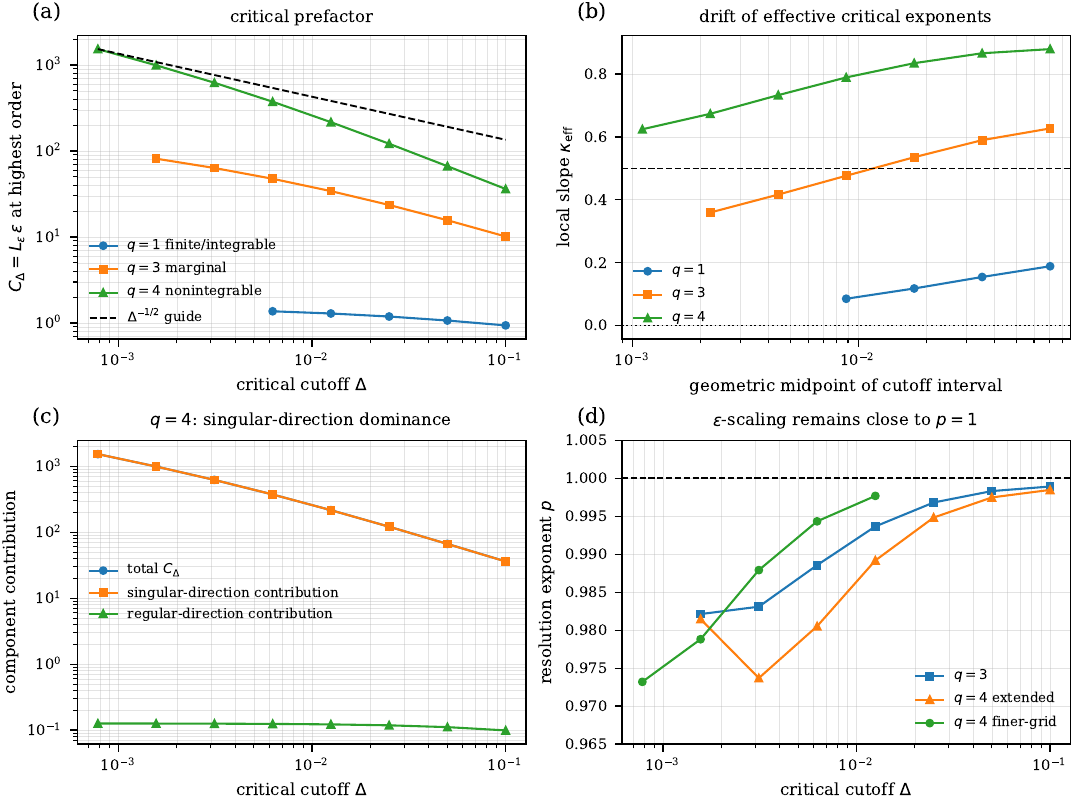}
    \caption{\HL_REF_AB{Finite coordinate-Hilbert diagnostics for the prescribed van der Waals response metric \(g_{vv}=s^{-q}\).  (a) Highest-order prefactors for \(q=1,3,4\), with a \(\Delta^{-1/2}\) guide.  (b) Finite-interval effective cutoff exponents.  (c) Directional decomposition for \(q=4\).  (d) Resolution exponents from the available grid orders.  The observed cutoff slopes remain finite-window quantities; the independent continuum and grid tests are reported separately.}}
    \label{fig:critical_integrability}
\end{figure*}

\section{\HL_REF_AB{Classical covering argument and regular benchmark details}}
\label{app:geometric_proof}
\HL_REF_AB{The material below is retained for completeness. The covering proof and regular benchmark formulas underpin the main-text geometric estimates and numerical tests.}
\subsection{Classical covering proof}
\begin{proof}
Let \(\Gamma_\eps=H_\eps([0,1])\) and let \(L=\Length_g[H_\eps]\).  Reparametrize the curve by arclength, so that it is defined on \([0,L]\).  Choose points \(p_k\) along the curve with arclength spacing at most \(\eps\).  The number of such points can be chosen to satisfy \(N\le \lceil L/\eps\rceil+1\le L/\eps+2\).  Every point of \(\Gamma_\eps\) lies within distance \(\eps\) of at least one of the points \(p_k\).  Therefore the \(\eps\)-neighborhood of \(\Gamma_\eps\) is contained in the union of balls of radius \(2\eps\):
\begin{equation}
    \Tube_\eps(\Gamma_\eps)
    \subseteq \bigcup_{k=1}^N B_g(p_k,2\eps).
    \label{eq:tube_cover_balls}
\end{equation}
Since \(H_\eps\) is \(\eps\)-space-filling, \(\M\subseteq \Tube_\eps(\Gamma_\eps)\).  Using subadditivity of volume and Eq.~\eqref{eq:ball_growth},
\begin{align}
    V_g
    &\le \Vol_g\bigl(\Tube_\eps(\Gamma_\eps)\bigr) \\
    &\le \sum_{k=1}^{N}\Vol_g\bigl(B_g(p_k,2\eps)\bigr) \\
    &\le \left(\frac{L}{\eps}+2\right)C_{\rm ball}(2\eps)^d \\
    &=2^d C_{\rm ball}\left(L\eps^{d-1}+2\eps^d\right).
    \label{eq:proof_volume_bound}
\end{align}
Solving for \(L\) gives Eq.~\eqref{eq:length_lower_bound}.  Multiplying by \(\eps^{d-1}\) and taking the lower limit proves Eq.~\eqref{eq:liminf_bound}.
\end{proof}
\begin{remark}[Relation to tube formulas]
The proof above uses only ball-volume growth and is therefore insensitive to self-intersections or corners of the finite Hilbert approximants.  For an embedded smooth curve \(\Gamma\) with positive reach, contained away from boundary effects and considered at tube radii below the reach, the Weyl tube expansion gives the sharper leading asymptotic form \cite{Weyl1939,Gray2004,Mattila1995}
\begin{equation}
    \Vol_g(\Tube_\eps(\Gamma))
    =\omega_{d-1}\eps^{d-1}\Length_g(\Gamma)+O(\eps^d),
    \label{eq:weyl_curve_tube}
\end{equation}
where \(\omega_{d-1}\) is the Euclidean volume of the unit ball in \(\R^{d-1}\); endpoint and boundary corrections enter at lower order under these assumptions.  If such a non-self-overlapping tube covers a volume \(V_g\), Eq.~\eqref{eq:weyl_curve_tube} would formally give the sharper leading constant \(V_g/\omega_{d-1}\) in the lower bound for \(\Length_g(\Gamma)\).  The main theorem deliberately sacrifices this constant by using the coarser ball-covering estimate, because finite-resolution Hilbert-, Peano-, or Morton-type traversals naturally have corners, close self-approaches, and sometimes self-intersections for which reach-based tube constants are not uniformly controlled.  The universal exponent \(1-d\), rather than the optimal tube constant, is the feature needed for exhaustive thermodynamic traversability.
\end{remark}
\subsection{Ideal gas}
\label{sm:ideal_metric_details}

For a single-component ideal gas at fixed particle number, use molar or per-particle variables \(x=(u,v)\), where \(u>0\) is internal energy per particle and \(v>0\) is volume per particle.  In dimensionless units, the entropy per particle may be written
\begin{equation}
    \sigma(u,v)=\frac{s(u,v)}{k_B}
    =\sigma_0+c\ln u+\ln v,
    \label{eq:ideal_entropy}
\end{equation}
where \(c=C_V/(Nk_B)>0\) is the dimensionless heat capacity per particle at constant volume.  \HL_REF_A{The Ruppeiner metric and flat coordinates are given in Eqs.~\eqref{eq:ideal_metric} and~\eqref{eq:ideal_flat_coords} of the main text.}

Eq.~\eqref{eq:ideal_metric} becomes Euclidean: \(\dd\ell^2=(\dd y^1)^2+(\dd y^2)^2\).  Thus the ideal-gas state space provides a flat benchmark in which the lower-bound scaling can be checked without curvature effects.

For a compact window
\begin{equation}
    \begin{aligned}
    \M_{\rm id}&=[u_-,u_+]\times[v_-,v_+],\\
    &0<u_-<u_+,\qquad 0<v_-<v_+,
    \end{aligned}
    \label{eq:ideal_window}
\end{equation}
its thermodynamic area is
\begin{equation}
    \Vol_{g^{\rm id}}(\M_{\rm id})
    =\sqrt{c}\,\ln\frac{u_+}{u_-}\,\ln\frac{v_+}{v_-}.
    \label{eq:ideal_volume}
\end{equation}
A Hilbert traversal uniform in the flat coordinates \((y^1,y^2)\), rather than in \((u,v)\), is therefore uniform with respect to the Ruppeiner volume element.  For \(d=2\), the lower bound reduces to \(\Length_{g^{\rm id}}[H_\eps]\ge C\eps^{-1}-O(\eps)\), with \(C\) proportional to the thermodynamic area of the chosen window.

\subsection{van der Waals fluid}
\label{sm:vdw_metric_details}

A van der Waals fluid supplies a curved state-space example and introduces singular sets associated with mechanical instability.  In dimensionless per-particle variables, let
\begin{equation}
    q(u,v)=u+\frac{a}{v},
    \qquad
    r(v)=v-b,
    \label{eq:vdw_qr}
\end{equation}
with \(a>0\), \(b>0\), \(q>0\), and \(r>0\).  A convenient entropy representation is
\begin{equation}
    \sigma(u,v)=\sigma_0+c\ln q(u,v)+\ln r(v).
    \label{eq:vdw_entropy}
\end{equation}
The Ruppeiner metric \(g^{\rm vdW}_{ij}=-\partial_i\partial_j\sigma\) has components
\begin{align}
    g^{\rm vdW}_{uu}
    &=\frac{c}{q^2},
    \label{eq:vdw_metric_uu}\\
    g^{\rm vdW}_{uv}
    &=-\frac{ca}{v^2q^2},
    \label{eq:vdw_metric_uv}\\
    g^{\rm vdW}_{vv}
    &=\frac{1}{r^2}-\frac{2ca}{v^3q}+\frac{ca^2}{v^4q^2}.
    \label{eq:vdw_metric_vv}
\end{align}
The admissible regular state space is chosen as a compact subset
\begin{equation}
    \M_{\rm vdW}\Subset\{(u,v):q(u,v)>0,\ v>b,\ g^{\rm vdW}(u,v)>0\},
    \label{eq:vdw_regular_domain}
\end{equation}
where \(g^{\rm vdW}>0\) denotes positive definiteness of the Hessian metric.  Spinodal or critical loci appear where the stability matrix loses positive definiteness.  They are excluded in the regular theorem and approached in Sec.~\ref{subsec:critical_integrability_results} through cutoff-dependent critical windows.

\HL_REF_AB{The numerical Ruppeiner benchmark uses the pullback to \((T,v)\), Eq.~\eqref{eq:vdw_ruppeiner_pullback_analysis}, and the specified coordinate-grid estimator.  A different, untested sampling choice would be a traversal approximately uniform with respect to the Ruppeiner volume element}
\begin{equation}
    \dd\mu_{g^{\rm vdW}}
    =\sqrt{g^{\rm vdW}_{uu}g^{\rm vdW}_{vv}-(g^{\rm vdW}_{uv})^2}\,\dd u\dd v.
    \label{eq:vdw_volume_element}
\end{equation}
\HL_REF_A{The length lower bound applies to either construction.  Rebuilding a uniform grid in different coordinates changes the traversal family and can change its prefactor even on the same metric window.  The Riemannian volume of that fixed window is unchanged by a passive coordinate transformation.}

\bibliography{references}
\end{document}